\documentclass[journal,onecolumn,11pt,draftclsnofoot]{IEEEtran}
\usepackage{hyperref}
\usepackage{tikz,pgfplots,tikzpeople}
\usetikzlibrary{positioning} 
\usepackage{multicol}
\usepackage{mdframed}
\usepackage{color}
\usepackage{caption,subcaption}
\usepackage{epsfig} 
\usepackage{amssymb, amsmath}
\usepackage{mathtools}
\usepackage{color} 
\usepackage{bm}
\usepackage{amsthm}
\usepackage{enumitem}
\usepackage{authblk}
\usepackage{cite}
\usepackage{wasysym}
\usepackage{arydshln}
\usepackage{graphicx}
\usepackage{color}

\usepackage{amssymb}
\usepackage{amsmath,amsfonts,amssymb}
\usepackage{verbatim}
\usepackage{stfloats}
\usepackage[bookmarks=false]{}

\newtheorem{theorem}{Theorem}

\newtheorem{definition}{Definition}
\newtheorem{lemma}{Lemma}

\newtheorem{remark}{Remark}

\newtheorem{proposition}{Proposition}

\usetikzlibrary{shapes,decorations,arrows,calc,arrows.meta,fit,positioning,patterns}

\allowdisplaybreaks

\DeclareMathOperator{\polylog}{polylog}
\DeclareMathOperator{\diag}{diag} 

\DeclareMathOperator{\vspan}{span}
\DeclareMathOperator{\lcm}{lcm}

\DeclarePairedDelimiter\ceil{\lceil}{\rceil}

\title{$X$-Secure $T$-Private Federated Submodel Learning with Elastic Dropout Resilience}
\author{Zhuqing Jia and Syed A. Jafar}
\affil{Center for Pervasive Communications and Computing (CPCC), UC Irvine\\
Email: \{zhuqingj, syed\}@uci.edu}
\date{}

\begin{document}

\maketitle
\begin{abstract}
Motivated by recent interest in   federated submodel learning, this work explores the fundamental problem of privately reading from and writing to a database comprised of $K$ files (submodels) that are stored across $N$ distributed servers according to an $X$-secure threshold secret sharing scheme. One after another, various users wish to retrieve their desired file, locally process the information and then update the file in the distributed database while keeping the identity of their desired file private from any set of up to $T$ colluding servers. The availability of servers changes over time, so  elastic dropout resilience is required. The main contribution of this work is an adaptive scheme, called ACSA-RW, that takes advantage of all currently available servers to reduce its communication costs, fully updates the database after each write operation even though the database is only partially accessible due to server dropouts, and ensures a memoryless operation of the network in the sense that the storage structure is preserved and future users may remain oblivious of the past history of  server dropouts. The ACSA-RW construction builds upon cross-subspace alignment (CSA) codes that were originally introduced for $X$-secure $T$-private information retrieval and have been shown to be natural solutions for secure distributed matrix multiplication problems.  ACSA-RW achieves the desired private read and write functionality with elastic dropout resilience, matches the best results for private-read from PIR literature, improves significantly upon available baselines for private-write, reveals a striking symmetry between upload and download costs, and exploits redundant storage dimensions to accommodate arbitrary read and write dropout servers up to certain threshold values. It also answers in the affirmative an open question by Kairouz et al. by exploiting synergistic gains from the joint design of private read and write operations.
\end{abstract}

\section{Introduction}
The rise of machine learning is marked by fundamental tradeoffs between  competing concerns. Central to this work are  1) the need for abundant training data, 2) the need for privacy, and 3) the need for low communication cost.  Federated learning \cite{McMahan_FL, FLsurvey1, FLsurvey2, FLsurvey3} is a distributed machine learning paradigm that addresses the first two concerns by allowing  distributed users/clients (e.g., mobile phones) to collaboratively train a shared model that is stored in a cluster of databases/servers (cloud) while keeping their training data private. The users retrieve the current model, train the model locally with their own training data, and then aggregate the modifications as  focused updates.  Thus, federated learning  allows utilization of abundant training data while preserving its privacy. However, this incurs higher communication cost as each update involves communication of all model parameters.

Communicating the full model may be unnecessary for large scale machine learning tasks where each user's local data is primarily relevant to a small part of the overall model. Federated \emph{Submodel} Learning (FSL) \cite{SFSM}  builds on this observation by partitioning the model into multiple submodels and allowing users to selectively train and update the submodels that are most relevant to their local view. This is the case, for example, in  the binary relevance method for multi-label classification\cite{Classifier_Read,Binary_Luaces}, which {\it independently} trains a series of binary classifiers (viewed as submodels), one for each label. Given a sample to be predicted, the compound model predicts all labels for which the respective classifiers yield a positive result. 

What makes federated submodel learning challenging is the  privacy constraint. The identity of the submodel that is being retrieved and updated by a user must remain private. Prior works \cite{SFSM, FLDP1, FLDP2, FLDP3, FLDP4}  that assume centralized storage of all submodels are generally able to provide relatively weaker privacy guarantees such as plausible deniability through  differential privacy mechanisms that perturb the data and rely on secure inter-user peer-to-peer communication for secure aggregation. On the other hand, it is noted recently by Kim and Lee in \cite{Kim_Lee_FSL} that if the servers that store the submodels are  \emph{distributed}, then stronger information theoretic guarantees such\footnote{By \emph{perfect} privacy we mean that absolutely no information is leaked about the identity of a user's desired submodel to any set of colluding servers up to a target threshold.} as ``perfect privacy'' may be attainable, without the need for user-to-user communication.  Indeed, in this work we focus on this setting of  distributed servers and perfect privacy. The challenge of federated submodel learning in this setting centers around three key questions.
\begin{enumerate}
    \item[{\bf Q1}] {\bf Private Read:} How can a user efficiently retrieve the desired submodel from the distributed servers without revealing which submodel is being retrieved?
    \item[{\bf Q2}] {\bf Private Write:} How can a user efficiently update the desired submodel to the distributed servers without revealing which submodel is being updated?
    \item[{\bf Q3}] {\bf Synergy of Private Read-Write:} Are there  synergistic gains in the \emph{joint} design of  retrieval and  update operations, and if so, then how to exploit these synergies?
\end{enumerate}
The significance of these fundamental questions goes well beyond federated submodel learning. As recognized by \cite{SFSM} the private read question (Q1) by itself is equivalent to the problem of Private Information Retrieval (PIR) \cite{PIRfirst, PIRfirstjournal}, which has recently been studied extensively from an information theoretic perspective \cite{Sun_Jafar_PIR,Sun_Jafar_TPIR,Sun_Jafar_SPIR,Wang_Sun_Skoglund,Wang_Skoglund_SPIREve,Banawan_Ulukus_BPIR,Tajeddine_Gnilke_Karpuk_Hollanti,Wang_Sun_Skoglund_BSPIR,Tajeddine_Gnilke_Karpuk, Wang_Skoglund_TSPIR, FREIJ_HOLLANTI, Sun_Jafar_MDSTPIR, Jia_Jafar_MDSXSTPIR,Zhou_Tian_Sun_Liu_Min_Size,Yang_Shin_Lee,Jia_Sun_Jafar_XSTPIR,Banawab_Arasli_Wei_Ulukus_Heterogeneous,Wei_Arasli_Banawan_Ulukus_Decentralized,Attia_Kumar_Tandon,Wei_Banawan_Ulukus, Tandon_CachePIR, Li_Gastpar_SSMUPIR,Wei_Banawan_Ulukus_Side,Chen_Wang_Jafar_Side,Sun_Jafar_MPIR,Yao_Liu_Kang_Multiround,Banawan_Ulukus_MPIR, Shariatpanahi_Siavoshani_Maddah,Wang_Banawan_Ulukus_PSI,Tian_Sun_Chen_Upload,Sun_Jafar_PC, Mirmohseni_Maddah,Chen_Wang_Jafar_Search,Jia_Jafar_GXSTPIR,Lu_Jia_Jafar_DBTPIR}.
Much less is known about  Q2 and Q3, i.e., the fundamental limits of private-write, and joint read-write solutions from the information theoretic perspective.  Notably, Q3 has also been highlighted previously as an open problem by Kairouz et al in \cite{FLsurvey1}. 

The problem of privately reading and writing data from a distributed memory falls under the larger umbrella of Distributed Oblivious RAM (DORAM)\cite{Lu_DORAM} primitives in theoretical computer science and cryptography. With a few limited exceptions (e.g., a specialized $4$-server construction in  \cite{Kushilevitz_DORAM_Sub} that allows information theoretic privacy), prior studies of DORAM generally take a cryptographic perspective, e.g., privacy is guaranteed subject to computational hardness assumptions, and the number of memory blocks is assumed to be much larger than the size of each block. In contrast, the  focus of this work is on Q2 and Q3  under the stronger notion of information theoretic privacy. Furthermore,  because our motivation comes from federated submodel learning, the size of a submodel is assumed to be significantly larger than the number of submodels (see motivating examples in \cite{SFSM} and Section \ref{sec:numerical}). Indeed, this is a pervasive assumption in the growing body of literature on information theoretic PIR\cite{Sun_Jafar_PIR,Sun_Jafar_TPIR,Sun_Jafar_SPIR,Wang_Sun_Skoglund,Wang_Skoglund_SPIREve,Banawan_Ulukus_BPIR,Tajeddine_Gnilke_Karpuk_Hollanti,Wang_Sun_Skoglund_BSPIR,Tajeddine_Gnilke_Karpuk, Wang_Skoglund_TSPIR, FREIJ_HOLLANTI, Sun_Jafar_MDSTPIR, Jia_Jafar_MDSXSTPIR,Zhou_Tian_Sun_Liu_Min_Size,Yang_Shin_Lee,Jia_Sun_Jafar_XSTPIR,Banawab_Arasli_Wei_Ulukus_Heterogeneous,Wei_Arasli_Banawan_Ulukus_Decentralized,Attia_Kumar_Tandon,Wei_Banawan_Ulukus, Tandon_CachePIR, Li_Gastpar_SSMUPIR,Wei_Banawan_Ulukus_Side,Chen_Wang_Jafar_Side,Sun_Jafar_MPIR,Yao_Liu_Kang_Multiround,Banawan_Ulukus_MPIR, Shariatpanahi_Siavoshani_Maddah,Wang_Banawan_Ulukus_PSI,Tian_Sun_Chen_Upload,Sun_Jafar_PC, Mirmohseni_Maddah,Chen_Wang_Jafar_Search,Jia_Jafar_GXSTPIR,Lu_Jia_Jafar_DBTPIR}. In a broad sense, our problem formulation in this paper is  motivated by applications of (information theoretic) PIR that also require private writes. There is no shortage of such applications, e.g., a distributed database of medical records that not only allows a physician to privately download the desired record (private read) but also to update the record with new information (private write), or a banking service that would similarly allow private reads and writes of financial records from authorized entities. Essentially, while FSL serves as our nominal application of interest based on prior works that motivated this effort, our problem formulation is broad enough to capture various distributed file systems that enable the users to read and write files without revealing the identity of the target file. The files are viewed as submodels, and the assumption that the size of the file is significantly larger than the number of the files captures the nature of file systems that are most relevant to this work.

{\it Overview:} We consider the federated submodel learning setting where the global model is partitioned into $K$ submodels, and stored among $N$ distributed servers according to an $X$-secure threshold secret sharing scheme, i.e., any set of up to $X$ colluding servers can learn nothing about the stored models, while the full model can be recovered from the data stored by any $X+K_c$ servers. One at a time, users update the submodel most relevant to their local training data. The updates must be $T$-private, i.e., any set of up to $T$ colluding servers must not learn anything about which submodel is being updated. The contents of the updates must be $X_{\Delta}$-secure, i.e., any set of up to $X_{\Delta}$ colluding servers must learn nothing about the contents of the submodel updates. The size of a submodel is significantly larger than the number of submodels, which is significantly larger than $1$, i.e., $L\gg K\gg 1$ where $L$  is the size of a submodel and $K$ is the number of submodels. Due to uncertainties of the servers' I/O states, link states, etc., an important concern in distributed systems is to allow resilience against servers that may temporarily drop out \cite{Kadhe_Fastsecagg,Lee_Lam_Pedarsani, Tandon_Lei_Dimakis_Karampatziakis,Chen_Wang_Charles_Papailiopoulos, Sohn_Han_Choi_Moon}. To this end, we assume that at each time $t, t\in\mathbb{N}$, a subset of  servers may be unavailable. These unavailable servers are referred to as read-dropout servers or write-dropout servers depending on whether the user intends to perform the private read or the private write operation. Since the set of dropout servers changes over time, and is assumed to be known to the user, the private read and write schemes must \emph{adapt} to the set of currently available servers. Note that this is different from the problem of stragglers in massive distributed computing applications where the set of responsive servers is not known in advance, because servers may become unavailable \emph{during} the lengthy time interval required for their  local computations. Since our focus is not on massive computing applications, the server side processing needed for private read and write is not as time-consuming. So the availabilities, which are determined in advance by the user before initiating the read or write operation, e.g., by pinging the servers, are not expected to change during the read or write operation. We do allow the server availabilities to change between the read and write operations due to the delay introduced by the intermediate processing that is needed at the user to generate his updated submodel. A somewhat surprising aspect of private write with unavailable servers is that even though the data at the unavailable servers cannot be updated, the collective storage at all servers (including the unavailable ones) must represent the updated models. The  redundancy in coded storage and the $X$-security constraints which require that the stored information at any $X$ servers is independent of the data, are essential in this regard.

Since the private-read problem (Q1) is essentially a form of PIR, our starting point is the $X$-secure $T$-private information retrieval scheme (XSTPIR) of \cite{Jia_Sun_Jafar_XSTPIR}. In particular, we build on the idea of cross-subspace alignment (CSA) from \cite{Jia_Sun_Jafar_XSTPIR}, and introduce a new private read-write scheme, called Adaptive CSA-RW (ACSA-RW) as an answer to Q1 and Q2. To our knowledge ACSA-RW is the first \emph{communication-efficient} private federated submodel learning scheme that achieves information-theoretically \emph{perfect} privacy. ACSA-RW also answers Q3 in the affirmative as it exploits  query structure from the private-read operation to reduce the communication cost for the private-write operation. The  evidence of synergistic gain in ACSA-RW from a joint design of submodel retrieval and submodel aggregation addresses the corresponding open problem highlighted in Section 4.4.4 of \cite{FLsurvey1}. The observation that the ACSA-RW scheme takes  advantage of  storage redundancy for private read and private write is indicative of  fundamental tradeoffs between  download cost, upload cost, data security level, and storage redundancy for security and recoverability. In particular, the storage redundancy for $X$-security is exploited by private write, while the storage redundancy for robust  recoverability is used for private read (see Theorem \ref{thm:acsarw} and Section \ref{sec:thresholds} for details). It is also remarkable that the ACSA-RW scheme requires absolutely no user-user communication, even though the server states change over time and the read-write operations are adaptive. In other words, a user is not required to be aware of the history of previous updates and the previous availability states of the servers (see Section \ref{sec:howtowrite} for details). The download cost and the upload cost achieved by each user is an increasing function of the number of unavailable servers at the time. When more servers are available, the download cost and the upload costs are reduced, which provides elastic dropout resilience (see Theorem \ref{thm:acsarw}). To this end, the ACSA-RW scheme uses adaptive MDS-coded answer strings. This idea originates from  CSA code constructions for the problem of coded distributed batch computation\cite{Jia_Jafar_CDBC}. In terms of comparisons against available baselines, we note (see Section \ref{sec:kim}) that  ACSA-RW improves significantly in both the communication efficiency and the level of privacy compared to \cite{Kim_Lee_FSL}. In fact, ACSA-RW achieves asymptotically optimal download cost when $X\geq X_{\Delta}+T$, and is  order-wise optimal in terms of the upload cost.  Compared with the $4$ server construction of information theoretic DORAM in \cite{Kushilevitz_DORAM_Sub}, (where $X=1, T=1, X_{\Delta}=0, N=4$)  ACSA-RW has better communication efficiency (the assumption of $L\gg K$ is important in this regard). For example, as the ratio $L/K$ approaches infinity, ACSA-RW achieves total communication cost (i.e., the summation of the download cost and the upload cost, normalized by the submodel size) of $6$, versus the communication cost of $8$ achieved by the construction in \cite{Kushilevitz_DORAM_Sub}.

{\it Notation: }  Bold symbols are used to denote vectors and matrices, while calligraphic symbols denote sets. By convention, let the empty product be the multiplicative identity, and the empty sum be the additive identity. For two positive integers $M, N$ such that $M\leq N$, $[M:N]$ denotes the set $\{M,M+1,\cdots,N\}$. We use the shorthand notation $[N]$ for $[1:N]$. $\mathbb{N}$ denotes the set of positive integers $\{1,2,3,\cdots\}$, and $\mathbb{Z}^*$ denotes the set $\mathbb{N}\cup\{0\}$. For a subset of integers $\mathcal{N}\subset\mathbb{N}$, $\mathcal{N}(i), i\in[|\mathcal{N}|]$ denotes its $i^{th}$ element, sorted in ascending order. The notation $\diag(\mathbf{D}_1,\mathbf{D}_2,\cdots,\mathbf{D}_n)$ denotes the block diagonal matrix, i.e., the main-diagonal blocks are square matrices $(\mathbf{D}_1,\mathbf{D}_2,\cdots,\mathbf{D}_n)$ and all off-diagonal blocks are zero matrices. For a positive integer $K$, $\mathbf{I}_K$ denotes the $K\times K$ identity matrix. For two positive integers $k,K$ such that $k\leq K$, $\mathbf{e}_K(k)$ denotes the $k^{th}$ column of the $K\times K$ identity matrix. The notation $\widetilde{\mathcal{O}}(a\log^2 b)$ suppresses\footnote{There is another standard definition of the notation $\widetilde{\mathcal{O}}$ which fully suppresses polylog terms, i.e, $\mathcal{O}(a\polylog(b))$ is represented by $\widetilde{\mathcal{O}}(a)$, regardless of the exact form of $\polylog(b)$. The definition used in this paper emphasizes the dominant factor in the polylog term.} polylog terms. It may be replaced with $\mathcal{O}(a\log^2 b)$ if the field supports the Fast Fourier Transform (FFT), and with  $\mathcal{O}(a\log^2b\log\log(b))$ if it does not\footnote{If the FFT is not supported by the field, Schönhage--Strassen algorithm\cite{Schonhage_Strassen_Schnelle} can be used for fast algorithms that require convolutions, with an extra factor of $\log\log b$ in the complexity.}.

\begin{figure}[!htbp]
    \centering
    \begin{subfigure}{1\columnwidth}
        \centering
        \begin{tikzpicture}[xscale=0.8,yscale=1]
            \node [draw, cylinder,aspect=0.2,shape border rotate=90,fill=teal!10, text=black, inner sep =0cm,minimum width=2.5cm] (S1) at (3.5cm, -1.5cm) {\footnotesize\begin{tabular}{c}Server $1$\\$\mathbf{S}_1^{(t-1)}$\end{tabular}};

            \node [draw, cylinder,aspect=0.2,shape border rotate=90,fill=teal!10, text=black, inner sep =0cm,minimum width=2.5cm] (S2) at (7.5cm, -1.5cm) {\footnotesize\begin{tabular}{c}Server $2$\\$\mathbf{S}_2^{(t-1)}$\end{tabular}};

            \node [rectangle, inner sep =0cm] (Ddots2) at (10cm, -1.5cm) {$\cdots$};

            \node [draw, cylinder,aspect=0.2,shape border rotate=90,fill=teal!10, text=black, inner sep =0cm,minimum width=2.5cm] (Si) at (12.5cm, -1.5cm) {\footnotesize\begin{tabular}{c}Server $i$\\$\mathbf{S}_i^{(t-1)}$\end{tabular}};

            \node [rectangle, inner sep =0cm] (Ddots2) at (15cm, -1.5cm) {$\cdots$};

            \node [draw, cylinder,aspect=0.2,shape border rotate=90,fill=teal!10, text=black, inner sep =0cm,minimum width=2.5cm] (SN) at (17.5cm, -1.5cm) {\footnotesize\begin{tabular}{c}Server $N$\\$\mathbf{S}_N^{(t-1)}$\end{tabular}};

            \node[alice, draw=black, text=black, minimum size=0.8cm, inner sep=0] (U) at (10.5cm, -5cm) {User $t$};

            \draw [olive!60!black, thick, arrows = {-Stealth[left]}] ([xshift=-10pt]U.north west) to node[below, sloped] {\footnotesize $Q_{1}^{(t,\theta_{t})}$} ([xshift=-10pt]S1.south);
            \draw [olive!60!black, thick, arrows = {-Stealth[left]}] ([xshift=-2pt]U.north) to node[below, sloped] {\footnotesize $Q_{2}^{(t,\theta_{t})}$} ([xshift=-2pt]S2.south);
              \draw [red, thick, dashed, arrows = {-Stealth[left]}] ([xshift=-2pt]U.north) to node[left] {\footnotesize $\O$} ([xshift=-2pt]Si.south);
            \draw [olive!60!black, thick, arrows = {-Stealth[right]}] ([xshift=10pt]U.north east) to node[below, sloped] {\footnotesize $Q_{N}^{(t,\theta_{t})}$} ([xshift=10pt]SN.south);

            \draw [red!60!black, thick, arrows = {-Stealth[left]}] (S1.south) to node[above, sloped] {\footnotesize $A_{1}^{(t,\theta_{t})}$} (U.north west);
            \draw [red!60!black, thick, arrows = {-Stealth[left]}] ([xshift=2pt]S2.south) to node[above, sloped] {\footnotesize $A_{n}^{(t,\theta_{t})}$} ([xshift=2pt]U.north);
            
            \draw [red, dashed, thick, arrows = {-Stealth[left]}] ([xshift=2pt]Si.south) to node[right] {\footnotesize $\O$} ([xshift=2pt]U.north);
            \draw [red!60!black, thick, arrows = {-Stealth[right]}] (SN.south) to node[above, sloped] {\footnotesize $A_{N}^{(t,\theta_{t})}$} (U.north east);

            \node[below=0.7cm of U, minimum size=0.3cm, inner sep=0.1cm] (P) {\small$\mathbf{W}_{\theta_t}^{(t-1)}=\left[W_{\theta_t}^{(t-1)}(1),W_{\theta_t}^{(t-1)}(2),\cdots,W_{\theta_t}^{(t-1)}(L)\right]^{\mathsf{T}}$};
            \node[left=0.5cm of U, minimum size=0.3cm, inner sep=0.1cm] (Q) {\small$\left(\theta_t, \mathcal{Z}_U^{(t)}\right)$};
            \draw[->]([yshift=-2pt]U.base)--(P);
            \draw[->](Q)--(U);
        \end{tikzpicture}
        \caption{The private-read phase of robust XSTPFSL. The $i^{th}$ server is unavailable, $i\in\mathcal{S}_r^{(t)}$.}
        \label{fig:XSTPFSLread}
        \vspace{0.5cm}
    \end{subfigure}

    \begin{subfigure}{1\columnwidth}
        \centering
        \begin{tikzpicture}[xscale=0.8,yscale=1]
            \node [draw, cylinder,aspect=0.2,shape border rotate=90,fill=teal!10, text=black, inner sep =0cm,minimum width=2.5cm] (S1) at (3.5cm, -1.5cm) {\footnotesize\begin{tabular}{c}Server $1$\\$\mathbf{S}_1^{(t-1)}$\end{tabular}};

            \node [draw, cylinder,aspect=0.2,shape border rotate=90,fill=teal!10, text=black, inner sep =0cm,minimum width=2.5cm] (S2) at (7.5cm, -1.5cm) {\footnotesize\begin{tabular}{c}Server $2$\\$\mathbf{S}_2^{(t-1)}$\end{tabular}};

            \node [rectangle, inner sep =0cm] (Ddots2) at (10cm, -1.5cm) {$\cdots$};

            \node [draw, cylinder,aspect=0.2,shape border rotate=90,fill=teal!10, text=black, inner sep =0cm,minimum width=2.5cm] (Si) at (12.5cm, -1.5cm) {\footnotesize\begin{tabular}{c}Server $j$\\$\mathbf{S}_j^{(t-1)}$\end{tabular}};

            \node [rectangle, inner sep =0cm] (Ddots2) at (15cm, -1.5cm) {$\cdots$};

            \node [draw, cylinder,aspect=0.2,shape border rotate=90,fill=teal!10, text=black, inner sep =0cm,minimum width=2.5cm] (SN) at (17.5cm, -1.5cm) {\footnotesize\begin{tabular}{c}Server $N$\\$\mathbf{S}_N^{(t-1)}$\end{tabular}};

            \node[alice, draw=black, text=black, minimum size=0.8cm, inner sep=0] (U) at (10.5cm, -5cm) {User $t$};

            \draw [red!60!black, thick, arrows = {-Stealth[left]}] (U.north west) to node[below, sloped] {\footnotesize $P_{1}^{(t,\theta_{t})}$} (S1.south);
            \draw [red!60!black, thick, arrows = {-Stealth[left]}] (U.north) to node[below, sloped] {\footnotesize $P_{2}^{(t,\theta_{t})}$} (S2.south);
            \draw [red, thick, dashed, arrows = {-Stealth[left]}] (U.north) to node[left] {\footnotesize $\O$} (Si.south);
            \draw [red!60!black, thick, arrows = {-Stealth[right]}] (U.north east) to node[below, sloped] {\footnotesize $P_{N}^{(t,\theta_{t})}$} (SN.south);

            \node[below=0.7cm of U, minimum size=0.3cm, inner sep=0.1cm] (Q) {\small$\left(\theta_t, \mathbf{\Delta}_t, \left(A_{n}^{(t,\theta_t)}\right)_{n\in[N]}, \mathcal{Z}_U^{(t)}\right)$};
            \draw[->](Q)--([yshift=-2pt]U.base);
        \end{tikzpicture}
        \caption{The private-write phase of robust XSTPFSL. The $j^{th}$ server is unavailable, $j\in\mathcal{S}_w^{(t)}$. Note that Server $i$ and Server $j$ may be different servers.}
        \label{fig:XSTPFSLwrite}
    \end{subfigure}
    \caption{The two phases of robust $X$-Secure $T$-Private Federated Submodel Learning (XSTPFSL) with arbitrary realizations of unavailable servers.}
    \label{fig:XSTPFSL}
\end{figure}
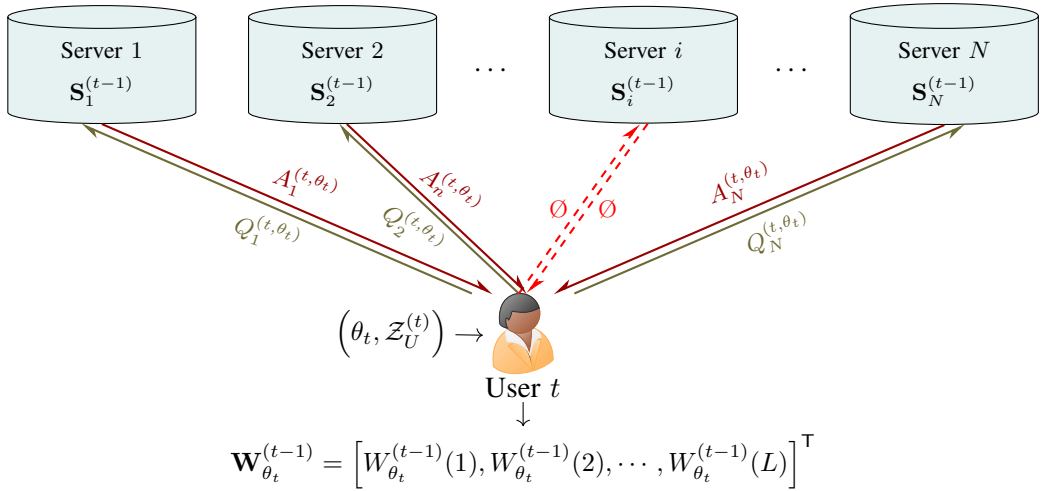
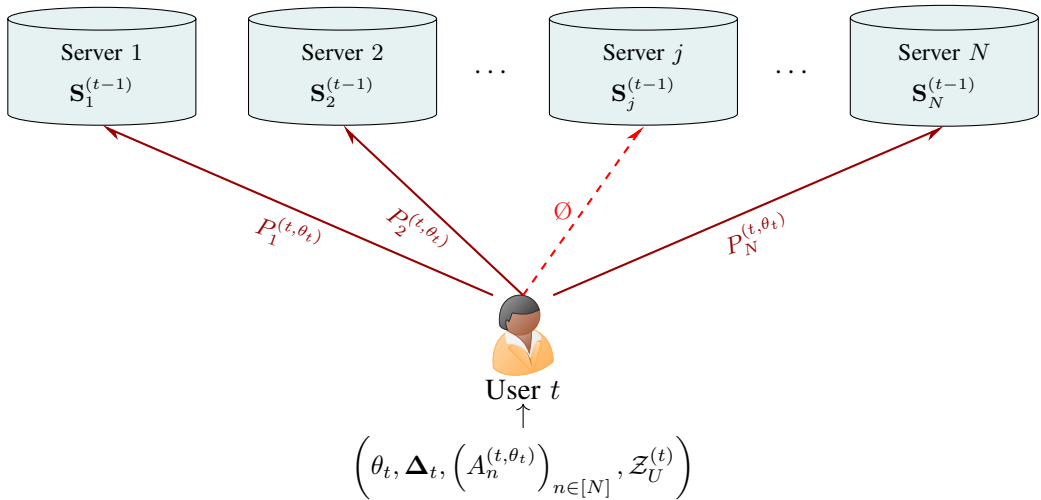
\section{Problem Statement: Robust XSTPFSL}\label{sec:probstat}
Consider $K$ initial submodels $\left(\mathbf{W}_1^{(0)}, \mathbf{W}_2^{(0)},\cdots,\mathbf{W}_K^{(0)}\right)$, each of which consists of $L$ uniformly i.i.d.\footnote{Note that the proposed ACSA-RW scheme does not require the assumption of uniformly i.i.d. model data to be correct, secure and private. The assumption is mainly used for converse arguments and communication cost metrics. However, we note that the uniformly i.i.d. model data assumption is in fact {\it not} very strong because submodel learning is performed locally, and it may be possible to achieve (nearly) uniformly i.i.d. model data by exploiting entropy encoding.} random symbols from a finite field\footnote{Note that the size of the finite field required is not too large. By our ACSA-RW scheme (see Theorem \ref{thm:acsarw}), it is sufficient to choose a finite field with $q\geq 2N$. Besides, finite field symbols are used as a representation of the submodels. Since the submodels are trained by the user {\it locally}, without loss of generality we can assume finite field representations.} $\mathbb{F}_q$. In particular, we have
\begin{align}
    H\left(\mathbf{W}_{1}^{(0)},\mathbf{W}_{2}^{(0)},\cdots,\mathbf{W}_{K}^{(0)}\right)=KL,
\end{align}
in $q$-ary units. Time slots are associated with users and their corresponding submodel updates, i.e., at time slot $t,t\in\mathbb{N}$, User $t$ wishes to perform the $t^{th}$ submodel update. At time $t, t\in\mathbb{Z}^*$, the $K$ submodels are denoted as $\left(\mathbf{W}_1^{(t)}, \mathbf{W}_2^{(t)},\cdots,\mathbf{W}_K^{(t)}\right)$. The submodels are represented as vectors, i.e., for all $t\in\mathbb{Z}^*$, $k\in[K]$,
\begin{align}
    \mathbf{W}_{k}^{(t)}=\left[W_{k}^{(t)}(1),W_{k}^{(t)}(2),\cdots,W_{k}^{(t)}(L)\right]^{\mathsf{T}}.
\end{align}
The $K$ submodels are  distributively stored among the $N$ servers. The storage at server $n, n\in[N]$ at time $t, t\in\mathbb{Z}^*$ is denoted as $\mathbf{S}_n^{(t)}$. Note that $\mathbf{S}_n^{(0)}$ represents the initial storage.

A full cycle of robust XSTPFSL is comprised of two phases --- the read phase and the write phase. At the beginning of the cycle,  User $t$ privately generates the desired index $\theta_t$, uniformly from $[K]$, and the user-side randomness $\mathcal{Z}_U^{(t)}$, which is intended to protect the user's privacy and security. In the read phase, User $t$ wishes to retrieve the submodel $\mathbf{W}_{\theta_t}^{(t-1)}$. At all times, nothing must be revealed about any (current or past) desired indices $(\theta_1,\theta_2,\cdots,\theta_t)$ to any set of up to $T$ colluding servers. To this end, User $t$ generates the read-queries $\left(Q_1^{(t,\theta_t)},Q_2^{(t,\theta_t)},\cdots,Q_N^{(t,\theta_t)}\right)$, where $Q_n^{(t,\theta_t)}$ is intended for the $n^{th}$ server, such that,
\begin{align}
    H\left(Q_n^{(t,\theta_t)}~\bigg|~ \theta_t, \mathcal{Z}_U^{(t)}\right) & =0, & \forall n\in[N].
\end{align}

Each of the currently available servers $n, n\in[N]\setminus\mathcal{S}_r^{(t)}$ is sent the query $Q_n^{(t,\theta_t)}$ by the user, and responds to the user with an answer $A_n^{(t,\theta_t)}$, such that,
\begin{align}
    H\left(A_n^{(t,\theta_t)}~\bigg|~ \mathbf{S}_{n}^{(t-1)}, Q_n^{(t,\theta_t)}\right) & =0, & \forall n\in[N]\setminus\mathcal{S}_r^{(t)}.
\end{align}
From the answers returned by the servers $n, n\in[N]\setminus\mathcal{S}_r^{(t)}$, the user must be able to reconstruct the desired submodel $\mathbf{W}_{\theta_t}^{(t-1)}$.
\begin{align}
     & \colorbox{black!10}{[Correctness]}\label{eq:correctness}\quad H\left(\mathbf{W}_{\theta_t}^{(t-1)} ~\bigg|~ \left(A_n^{(t,\theta_t)}\right)_{n\in[N]\setminus\mathcal{S}_r^{(t)}}, \left(Q_n^{(t,\theta_t)}\right)_{n\in[N]},\theta_t\right)=0.\notag
\end{align}
This is the end of the read phase.

Upon finishing the local submodel training, User $t$ privately generates an increment\footnote{In general, the increment is the difference between the new submodel and the old submodel. We note that no generality is lost by assuming additive increments because the submodel training is performed locally.} for the $\theta_t^{th}$ submodel. The increment is represented as a vector, $\mathbf{\Delta}_t=[\Delta^{(t)}_1,\Delta^{(t)}_2,\cdots,\Delta^{(t)}_L]^{\mathsf{T}}$, which consists of $L$ i.i.d. uniformly distributed symbols from the finite field $\mathbb{F}_q$, i.e., $H(\mathbf{\Delta}_t)=L$ in $q$-ary units. The increment $\mathbf{\Delta}_t$ is intended to update the $\theta_t^{th}$ submodel $\mathbf{W}_{\theta_t}^{(t-1)}$, such that the next user who wishes to make an update, User $t+1$, is able to retrieve the submodel $\mathbf{W}_{\theta_t}^{(t)}=\mathbf{W}_{\theta_t}^{(t-1)}+\mathbf{\Delta}_t$ if $\theta_{t+1}=\theta_{t}$, and the submodel $\mathbf{W}_{\theta_t'}^{(t)}=\mathbf{W}_{\theta_t'}^{(t-1)}$ if $\theta_{t+1}=\theta_{t}'\neq\theta_t$. In other words, for all $t\in\mathbb{N}$, $k\in[K]$, the submodel $\mathbf{W}_{k}^{(t)}$ is defined recursively as follows.
\begin{align}
    \mathbf{W}_{k}^{(t)}=\left\{\begin{array}{cc}
        \mathbf{W}_{k}^{(t-1)}+\mathbf{\Delta}_t & k=\theta_t,     \\
        \mathbf{W}_{k}^{(t-1)}                   & k\neq \theta_t.
    \end{array}\right.
\end{align}
User $t$ initializes the write phase by generating the write-queries $\left(P_1^{(t,\theta_t)}, P_2^{(t,\theta_t)}, \cdots, P_N^{(t,\theta_t)}\right)$. For ease of notation, let the write-queries be nulls for the write-dropout servers, i.e., $P_n^{(t,\theta_t)}=\O$ for $ n\in\mathcal{S}_w^{(t)}$. For all $n\in[N]$,
\begin{align}
    H\left(P_n^{(t,\theta_t)}~\bigg|~ \theta_t, \mathcal{Z}_U^{(t)},\left(A_{n}^{(t,\theta_t)}\right)_{n\in[N]},\mathbf{\Delta}_t\right) & =0.
\end{align}

The user sends the write-query $P_n^{(t,\theta_t)}$ to the $n^{th}$ server, $n\in[N]\setminus\mathcal{S}_w^{(t)}$, if the 
server was available in the read-phase and therefore already received the read-query. Otherwise, if the server was not available during the read phase, then the user sends\footnote{For ease of exposition we will assume in the description of the scheme that during the read phase, the read queries are sent to \emph{all} servers that are available during the read phase, or will become available later during the write phase, so that only the write-queries need to be sent during the write phase.}   both read and write queries $(Q_n^{(t,\theta_t)}, P_n^{(t,\theta_t)})$. 
Still, any set of up to $T$ colluding servers must learn nothing about the desired indices $(\theta_1,\theta_2,\cdots,\theta_t)$. 

Upon receiving the write-queries, each of the servers $n, n\in[N]\setminus\mathcal{S}_w^{(t)}$ updates its storage based on the existing storage $\mathbf{S}_{n}^{(t-1)}$ and the queries for the two phases $\left(P_{n}^{(t,\theta_t)},Q_{n}^{(t,\theta_t)}\right)$, i.e.,
\begin{align}
    H\left(\mathbf{S}_{n}^{(t)}~\bigg|~ \mathbf{S}_{n}^{(t-1)},P_{n}^{(t,\theta_t)},Q_{n}^{(t,\theta_t)}\right) & =0.
\end{align}
On the other hand, the write-dropout servers are unable to perform any storage update.
\begin{align}
    \mathbf{S}_{n}^{(t)}=\mathbf{S}_{n}^{(t-1)} & , & \forall n\in\mathcal{S}_w^{(t)}.
\end{align}

Next, let us formalize the security and privacy constraints. $T$-privacy  guarantees that at any time, any set of up to $T$ colluding servers learn nothing about the indices $(\theta_1,\theta_2,\cdots,\theta_t)$ from all the read and write queries and storage states.
\begin{align}
    \colorbox{black!10}{[$T$-Privacy]}\quad & I\left((\theta_{\tau})_{\tau\in[t]};\left(P_n^{(t,\theta_t)},Q_n^{(t,\theta_t)}\right)_{n\in\mathcal{T}}~\bigg|~\left(\mathbf{S}_{n}^{(\tau-1)},P_{n}^{(\tau-1,\theta_{\tau-1})},Q_{n}^{(\tau-1,\theta_{\tau-1})}\right)_{\tau\in[t],n\in\mathcal{T}}\right)=0,\notag \\
                                            & \forall \mathcal{T}\in[N],|\mathcal{T}|=T, t\in\mathbb{N},\label{eq:privacy}
\end{align}
where for all $n\in[N]$, we define $P_n^{(0,\theta_0)}=Q_n^{(0,\theta_0)}=\O$.

Similarly, any set of up to $X_{\Delta}$ colluding servers must learn nothing about the increments $(\mathbf{\Delta}_{1},\mathbf{\Delta}_{2},\cdots,\mathbf{\Delta}_{t})$.
\begin{align}
    \colorbox{black!10}{[$X_{\Delta}$-Security]}\quad & I\left(\left(\mathbf{\Delta}_{\tau}\right)_{\tau\in[t]};\left(P_n^{(t,\theta_t)}\right)_{n\in\mathcal{X}}~\bigg|~\left(\mathbf{S}_{n}^{(\tau-1)},P_{n}^{(\tau-1,\theta_{\tau-1})},Q_{n}^{(\tau,\theta_{\tau})}\right)_{\tau\in[t],n\in\mathcal{X}}\right)=0, \notag \\
                                                      & \forall \mathcal{X}\subset[N], |\mathcal{X}|=X_{\Delta}, t\in\mathbb{N}.
\end{align}

The storage at the $N$ servers is formalized according to a threshold secret sharing scheme. Specifically, the storage at any set of up to $X$ colluding servers must reveal nothing about the submodels. Formally,
\begin{align}
                                             & H\left(\mathbf{S}_n^{(0)}~\bigg|~\left(\mathbf{W}_{k}^{(0)}\right)_{k\in[K]}, \mathcal{Z}_S\right)=0, \forall n\in[N],                                                         \\
    \colorbox{black!10}{[$X$-Security]}\quad & I\left(\left(\mathbf{W}_{k}^{(t)}\right)_{k\in[K]};\left(\mathbf{S}_n^{(t)}\right)_{n\in\mathcal{X}}\right)=0, \forall \mathcal{X}\subset[N], |\mathcal{X}|=X, t\in\mathbb{Z}^*,
\end{align}
where for all $n\in[N]$, we define $\mathbf{S}_n^{-1}=\O$. Note that $\mathcal{Z}_S$ is the private randomness used by the secret sharing scheme that implements $X$-secure storage across the $N$ servers.

There is a subtle difference in the security constraint that we impose on the storage, and the previously specified security and privacy constraints on updates and queries. To appreciate this difference let us make a distinction between the notions of an internal adversary and an external adversary. We say that a set of colluding servers forms an internal adversary if those colluding servers have access to not only their current storage, but also their entire history of previous stored values and queries. Essentially the internal adversary setting represents a greater security threat because the servers themselves are dishonest and surreptitiously keep records of all their history in an attempt to learn from it. In contrast, we say that a set of colluding servers forms an external adversary if those colluding servers have access to only their current storage, but not to other historical information. Essentially, this represents an external adversary who is able to steal the current information from honest servers who do not keep records of their historical information. Clearly, an external adversary is weaker than an internal adversary. Now let us note that while the $T$-private queries and the $X_\Delta$-secure updates are protected against \emph{internal} adversaries, the $X$-secure storage is only protected against \emph{external} adversaries. This is mainly because we will generally assume $X>\max(X_\Delta, T)$, i.e., a higher security threshold for storage, than for updates and queries. Note that once the number of compromised servers exceeds $\max(X_\Delta,T)$, the security of updates and the privacy of queries is no longer guaranteed. In such settings the security of storage is still guaranteed, albeit in a weaker sense (against external adversaries). On the other hand if the number of compromised servers is small enough, then indeed secure storage may be guaranteed in a stronger sense, even against internal adversaries. We  refer the reader to Remark \ref{remark:extintadv}  for further insight into this aspect.

The independence among various quantities of interest is specified for all $t\in\mathbb{N}$ as follows.
\begin{align}
     & H\left(\left(\mathbf{W}_k^{(0)}\right)_{k\in[K]},\left(\mathbf{\Delta}_{\tau}\right)_{\tau\in[t]},\left(\theta_{\tau}\right)_{\tau\in[t]},\left(\mathcal{Z}_U^{(\tau)}\right)_{\tau\in[t]},\mathcal{Z}_S\right)\notag                                          \\
     & =H\left(\left(\mathbf{W}_k^{(0)}\right)_{k\in[K]}\right)+H\left(\left(\mathbf{\Delta}_{\tau}\right)_{\tau\in[t]}\right)+H\left(\left(\theta_{\tau}\right)_{\tau\in[t]}\right)+H\left(\left(\mathcal{Z}_U^{(\tau)}\right)_{\tau\in[t]}\right)+H(\mathcal{Z}_S).
\end{align}

To evaluate the performance of a robust XSTPFSL scheme, we consider the following metrics. The first two metrics focus on communication cost. For $t\in\mathbb{N}$, the download cost $D_t$ is the expected (over all realizations of queries) number of $q$-ary symbols downloaded by User $t$, normalized by $L$. The upload cost $U_t$ is the expected number of $q$-ary symbols uploaded by User $t$, also normalized by $L$. %
The next metric focuses on storage efficiency. For $t\in\mathbb{Z}^*$, the storage efficiency is defined as the ratio of the total data content to the total storage resources consumed by a scheme, i.e., $\eta^{(t)}=\frac{KL}{\sum_{n\in[N]} H\left(S_n^{(t)}\right)}$. If $\eta^{(t)}$ takes the same value for all $t\in\mathbb{Z}^*$, we use the compact notation $\eta$ instead.

\section{Main Result: The ACSA-RW Scheme for Private Read/Write}
The main contribution of this work is the ACSA-RW scheme, which allows private read and private write from $N$ distributed servers according to the problem statement provided in Section \ref{sec:probstat}. The scheme achieves storage efficiency $K_c/N$, i.e., it uses a total storage of $(KLN/K_c)\log_2 q$ bits across $N$ servers in order to store the $KL\log_2 q$ bits of actual data, where $K_c\in\mathbb{N}$, and allows arbitrary read-dropouts and write-dropouts as long as the number of dropout servers is less than the corresponding threshold values. The thresholds are defined below and their relationship to  redundant storage dimensions is explained in Section \ref{sec:thresholds}.
\begin{align}
\mbox{\bf Read-dropout threshold:} &&S_r^{\text{\normalfont thresh}}&\triangleq N-(K_c+X+T-1).\\
\mbox{\bf Write-dropout threshold:} &&S_w^{\text{\normalfont thresh}}&\triangleq X-(X_\Delta+T-1).
\end{align}
It is worth emphasizing that the ACSA-RW scheme  does not just tolerate  dropout servers, it adapts (hence the \emph{elastic} resilience) to the number of available servers so as to reduce its communication cost. The elasticity would be straightforward if the only concern was the private-read operation because known PIR schemes can be adapted to the number of available servers. What makes the elasticity requirement non-trivial is that the scheme must accommodate both private read and private write. The private-write requirements are particularly intriguing, almost paradoxical in that  the coded storage across all $N$ servers needs to be updated to be consistent with  the new submodels, even though some of those servers (the write-dropout servers) are unavailable, so their stored information cannot be changed. Furthermore, as server states continue to change over time, future updates need no knowledge of prior dropout histories. Also of interest are the tradeoffs between storage redundancy and the resources needed for private-read and private-write functionalities. Because many of these aspects  become more intuitively transparent when $L\gg K\gg 1$, the asymptotic setting is used to present the main result in Theorem \ref{thm:acsarw}.  In particular, by suppressing minor terms (which can be found in the full description of the scheme), the asymptotic setting  reveals an elegant symmetry between the upload and download costs. 
The remainder of this section is devoted to stating and then understanding the implications of Theorem \ref{thm:acsarw}. The scheme itself is presented in the form of the proof of Theorem \ref{thm:acsarw} in Section \ref{sec:acsarwproof}. 

\begin{theorem}\label{thm:acsarw}
In the limit $L/K\rightarrow \infty$, for all $t\in\mathbb{N}$
 the ACSA-RW scheme achieves the following download, upload cost pair $(D_t, U_t)$ and storage efficiency $\eta$:
\begin{align}
       (D_t, U_t)&=\left(\frac{N-|\mathcal{S}_r^{(t)}|}{S_r^{\text{\normalfont thresh}}-|\mathcal{S}_r^{(t)}|}, ~\frac{N-|\mathcal{S}_w^{(t)}|}{S_w^{\text{\normalfont thresh}}-|\mathcal{S}_w^{(t)}|}\right),&& \eta=\frac{K_c}{N}, \label{eq:ACSAachi}
\end{align}
for any $K_c\in\mathbb{N}$ such that $|\mathcal{S}_r^{(t)}|<S_r^{\text{\normalfont thresh}}$, $|\mathcal{S}_w^{(t)}|<S_w^{\text{\normalfont thresh}}$, and the  field size $q\geq N+\max\left\{S_r^{\text{\normalfont thresh}}, S_w^{\text{\normalfont thresh}}, K_c\right\}$.
\end{theorem}

\subsection{Observations}
\subsubsection{Storage Redundancy and Private Read/Write Thresholds}\label{sec:thresholds}
The relationship between redundant storage dimensions and the private read/write thresholds is conceptually illustrated in Figure \ref{fig:pertrade}.
          \begin{figure}[h]
              \centering
              \begin{tikzpicture}[xscale=0.7,yscale=0.9]

                  \draw[color=black,very thick,] (-1,0) rectangle (1,1.5) node[fill=white,pos=.5] {\footnotesize $K_c$};
                  \draw[color=black,very thick,pattern=north east lines,, pattern color=red!70!black] (1+10,0) rectangle (4+10,1.5) node[fill=white,pos=.5] {\footnotesize $T-1$};
                  \draw[color=black,very thick,pattern=north east lines, pattern color=red!70!black] (4+10,0) rectangle (6+10,1.5) node[fill=white,pos=.5] {\footnotesize $|\mathcal{S}_r^{(t)}|$};
                  \draw[color=black,very thick,pattern=north east lines, pattern color=red!70!black] (6+10,0) rectangle (10+10,1.5) node[fill=white,pos=.5] {\footnotesize $S_r^{\text{\normalfont thresh}}-|\mathcal{S}_r^{(t)}|$};
                  \draw[color=black,very thick,pattern=north west lines,pattern color=teal!70!black] (10-9,0) rectangle (12-9,1.5) node[fill=white,pos=.5] {\footnotesize $X_\Delta$};
                  \draw[color=black,very thick,pattern=north west lines,pattern color=teal!70!black] (12-9,0) rectangle (14-9,1.5) node[fill=white,pos=.5] {\footnotesize $T-1$};
                  \draw[color=black,very thick,pattern= north west lines,pattern color=teal!70!black] (14-9,0) rectangle (16-9,1.5) node[fill=white,pos=.5] {\footnotesize $|\mathcal{S}_w^{(t)}|$};
                  \draw[color=black,very thick,pattern=north west lines,,pattern color=teal!70!black] (16-9,0) rectangle (20-9,1.5) node[fill=white,pos=.5] {\footnotesize $S_w^{\text{\normalfont thresh}}-|\mathcal{S}_w^{(t)}|$};

                  \draw [
                      very thick,
                      decoration={
                              brace,
                              mirror,
                              raise=0.2cm
                          },
                      decorate
                  ] (-1,0) -- (20,0)
                  node [pos=0.5,anchor=north,yshift=-0.35cm] {\large$N$};
                  
                  \draw [
                      very thick,
                      decoration={
                              brace,
                              raise=0.2cm
                          },
                      decorate
                  ] (10-9,1.5) -- (20-9,1.5)
                  node [pos=0.5,anchor=north,yshift=1cm] {\large$X$};

              \end{tikzpicture}
              \caption{\small \it Conceptual partitioning of total server storage space ($N$ dimensions) into data content ($K_c$ dimensions), storage redundancy that is exploited by private-write ($X$ dimensions), and storage redundancy that is exploited by private-read ($N-K_c-X$ dimensions).}
              \label{fig:pertrade}
          \end{figure}
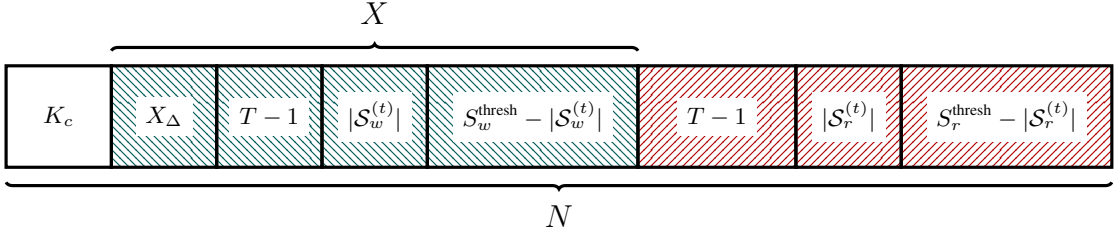
          
The total storage utilized by the ACSA-RW scheme across $N$ servers is represented as the overall $N$ dimensional space in Figure \ref{fig:pertrade}. Out of this, the actual data occupies only $K_c$ dimensions, which is why the storage efficiency of ACSA-RW is $K_c/N$. Because the storage must be $X$-secure, i.e., any set of up to $X$ colluding servers cannot learn anything about the data, it follows that out of the $N$ dimensions of storage space, $X$ dimensions are occupied by information that is independent of the actual data. The storage redundancy represented by these $X$ dimensions will be essential to enable the private-write functionality. But first let us consider the private-read operation for which we have a number of prior results on PIR as baselines for validating our intuition. From  Figure \ref{fig:pertrade} we note that outside the $X$ dimensions of redundancy that was introduced due to data-security, the $T$-privacy constraint adds a storage redundancy of another $T-1$ dimensions that is shown in red. To understand this, compare the asymptotic (large $K$) capacity of MDS-PIR \cite{Banawan_Ulukus}: $C_{\text{MDS-PIR}}=1-K_c/N$ with the (conjectured) asymptotic capacity of MDS-TPIR \cite{FREIJ_HOLLANTI, Sun_Jafar_MDSTPIR}: $C_{\text{MDS-TPIR}}=1-(K_c+T-1)/N$. To achieve a non-zero value of the asymptotic capacity, the former requires $N> K_c$, but the latter requires $N> K_c+T-1$. Equivalently, the former allows a read-dropout threshold of $N-K_c$ while the latter allows a read-dropout threshold of $N-(K_c+T-1)$. In fact, going further to the (conjectured) asymptotic capacity of MDS-XSTPIR \cite{Jia_Jafar_MDSXSTPIR}, which also includes the $X$-security constraint, we note that a non-zero value of capacity requires $N> K_c+X+T-1$. Intuitively, we may interpret this as: the $X$-security constraint increases the demands on  storage redundancy by $X$ dimensions and the $T$-privacy constraint increases the demands on storage redundancy by another $T-1$ dimensions. This is what is represented in Figure \ref{fig:pertrade}. Aside from the $K_c$ dimensions occupied by data,  the $X$ dimensions of redundancy added by the security constraint, and the $T-1$ dimensions of redundancy added by the $T$-privacy constraint, the remaining dimensions at the right end of the figure are used to accommodate read-dropouts. Indeed, this is what determines the read-dropout threshold, as we note from Figure \ref{fig:pertrade} that $N-(X+K_c+(T-1)) = S_r^{\text{\normalfont thresh}}$. Now consider the private-write operation which is novel and thus lacks comparative baselines. What is remarkable is the synergistic aspect of private-write, that it does not add further redundancy beyond the $X$ dimensions of storage redundancy already added by the $X$-security constraint. Instead, it operates within these $X$ dimensions to create further sub-partitions. Within these $X$ dimensions, a total of $X_\Delta+T-1$ dimensions are used to achieve $X_\Delta$-secure updates that also preserve $T$-privacy, and the remaining dimensions are used to accommodate write-dropout servers, giving us the write-dropout threshold as $X-(X_\Delta+T-1)=S_w^{\text{\normalfont thresh}}$. Remarkably in Figure \ref{fig:pertrade}, the $(K_c, X_\Delta, T-1, S_w^{\text{\normalfont thresh}})$ partition structure for private-write replicates at a finer level the original $(K_c, X, T-1,S_r^{\text{\normalfont thresh}})$ partition structure of private-read. Also remarkable is a new constraint introduced by the private-write operation that is not  encountered in prior works on PIR --- the feasibility of ACSA-RW requires $X\geq T$. Whether or not this constraint is fundamental in the asymptotic setting of large $K$ is an open problem for future work.

\subsubsection{Optimality} \label{sec:optimal}
Asymptotic (large $K$) optimality of ACSA-RW remains an open question in general. Any attempt to resolve this question runs into other prominent open problems in the information-theoretic PIR literature, such as the asymptotic capacity of MDS-TPIR \cite{FREIJ_HOLLANTI, Sun_Jafar_MDSTPIR} and MDS-XSTPIR \cite{Jia_Jafar_MDSXSTPIR} that also remain open. Nevertheless, it is worth noting that Theorem \ref{thm:acsarw} matches or improves upon the best known results in all cases where such results are available. In particular, the private-read phase of ACSA-RW scheme recovers a universally robust $X$-secure $T$-private information retrieval \cite{Bitar_StairPIR,Jia_Sun_Jafar_XSTPIR} scheme. When $K_c=1, X\geq X_{\Delta}+T$, it achieves the asymptotic capacity; and when $K_c>1, X\geq X_{\Delta}+T$, it achieves the conjectured asymptotic capacity \cite{Jia_Jafar_MDSXSTPIR}. While much less is known about optimal private-write schemes, it is clear that ACSA-RW significantly improves upon previous work as explained in Section \ref{sec:kim}. Notably, both upload and download costs are $O(1)$ in $K$, i.e., they do not scale with $K$. Thus, at the very least the costs are orderwise optimal. Another interesting point of reference is the best case scenario, where we have no dropout servers, $|\mathcal{S}_r^{(t)}|=0, |\mathcal{S}_w^{(t)}|=0$. The total communication cost of ACSA-RW in this case, i.e., the sum of upload and download costs, is $D_t+U_t = N\left(\frac{1}{S_r^{\text{thresh}}}+\frac{1}{S_w^{\text{thresh}}}\right)$, which is minimized when $K_c=1$, $S_r^{\text{thresh}}=S_w^{\text{thresh}}$,  and $X= (N+X_\Delta -1)/2$. For large $N$, we have $D_t+U_t\approx 2+2 = 4$, thus in the best case scenario, ACSA-RW is optimal within a factor of $2$. Finally, on a speculative note, perhaps the most striking aspect of Theorem \ref{thm:acsarw} is the symmetry between upload and download costs, which (if not coincidental) bodes well for their fundamental significance and information theoretic optimality.

\subsubsection{The Choice of Parameter $K_c$}\label{sec:kc}
The choice of the parameter $K_c$ in ACSA-RW determines the storage efficiency of the scheme, $\eta=K_c/N$. At one extreme, we have the smallest possible value of $K_c$, i.e., $K_c=1$, which is the least efficient storage setting, indeed the storage efficiency is analogous to replicated storage, each server uses as much storage space as the size of all data ($KL\log_2 q$ bits). This setting yields the best (smallest) download costs. The other extreme corresponds to the maximum possible value of $K_c$, which is obtained as $K_c=N-(X+T)$ because the dropout thresholds cannot be smaller than $1$. At this extreme, storage is the most efficient,  but there is no storage redundancy left to accommodate any read dropouts, and the download cost of ACSA-RW takes its maximal value, equal to $N$. Remarkably, the upload cost of the private-write operation does not depend on $K_c$. However, $K_c$ is  significant for another reason; it determines the access complexity (see Remark \ref{remark:ac}) of both private read and write operations, i.e., the number of bits that are read from or written to by each available server. In particular,  the access complexity of each available server in the private read or write phases is at most $(KL/K_c)\log_2 q$, so for example, increasing $K_c$ from $1$ to $2$ can reduce the access complexity in half, while simultaneously doubling the storage efficiency.

\subsubsection{Tradeoff between Upload and Download Costs}\label{sec:DtUt}
The trade-off between the upload cost and the download cost of the ACSA-RW scheme is illustrated via two examples in Figure \ref{fig:tradeoff}, where we have $N=10, X_{\Delta}=T=1$ for the blue solid curve, and $N=10, X_{\Delta}=1, T=2$ for the red solid curve. For both examples, we set $K_c=1$ and assume that there are no dropout servers. The trade-off is achieved with various choices of $X$. For the example shown in the blue solid curve, we set $X=(2,3,4,5,6,7,8)$. For the example in the red solid curve, we set $X=(3,4,5,6,7)$. Note that the most balanced trade-off point is achieved when $X=N/2$.
          \begin{figure}[h]
              \centering
              \begin{tikzpicture}
                  \begin{axis}[
                          grid style={line width=.1pt, draw=gray!20},
                          major grid style={line width=.2pt,draw=gray!50},
                          grid=both,
                          axis lines = left,
                          minor tick num=1,
                          xlabel = {Upload cost, $U_t\longrightarrow$},
                          ylabel = {Download cost, $D_t\longrightarrow$},
                          xmin=0,
                          xmax=11,
                          ymin=0,
                          ymax=11,
                          very thick,
                          use fpu=false
                      ]

                      \addplot [mark=*,
                          color=blue!60!black,
                          ultra thick,
                          samples at={2,3,4,5,6,7,8},
                      ] ({(10)/(x-1)},{10/(9-x)})  node[pos=0.01,above=-1.7em ,color=black]{\footnotesize ($X=2$)}  node[pos=0.87,left=1em ,color=black,rotate=-90]{\footnotesize ($X=8$)};

                      \addplot [mark=*,
                          color=red!60!black,
                          ultra thick,
                          samples at={3,4,5,6,7},
                      ] ({(10)/(x-2)},{10/(8-x)}) node[pos=0.55,above right,color=black,rotate=-30]{$\longleftarrow X$} node[pos=0.01,above=0.3em ,color=black]{\footnotesize ($X=3$)}  node[pos=1.05,right=1em ,color=black,rotate=-90]{\footnotesize ($X=7$)};

                  \end{axis}

              \end{tikzpicture}
              \caption{\small \it Upload, download costs pairs $(U_t,D_t)$ of the  ACSA-RW scheme in the asymptotic setting $L\gg K\gg 1$, for $N=10, X_\Delta=T=1$ (the blue curve) and $N=10, X_\Delta=1, T=2$ (the red curve), with various choices of $X$. Both examples  assume that $K_c=1$ and there are no dropout servers.}
              \label{fig:tradeoff}
          \end{figure}
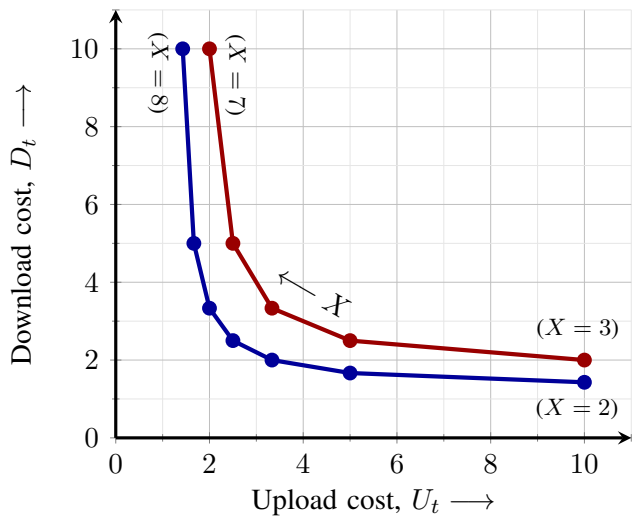

\subsubsection{Synergistic Gains from Joint Design of Private Read and Write}
A notable aspect of the ACSA-RW scheme, which answers in the affirmative an open question raised in Section 4.4.4 of \cite{FLsurvey1}, is the synergistic gain from the joint design of the read phase and the write phase. While the details of the scheme are non-trivial and can be found in Section \ref{sec:acsarwproof}, let us provide an intuitive explanation here by ignoring some of the details. As a simplification, let us ignore security constraints and consider the database represented by a vector $\mathbf{W}=[W_1, W_2, \cdots, W_K]^{\mathsf{T}}$ that consists of symbols from the $K$ submodels. Suppose the user is interested in $W_{\theta_t}$ for some index $\theta_t\in[K]$ that must be kept private. Note that $W_{\theta_t}=\mathbf{W}^{\mathsf{T}}\mathbf{e}_K(\theta_t)$, where $\mathbf{e}_K(\theta_t)$ is the standard basis vector, i.e., the desired symbol is obtained as an inner product of the database vector and the basis vector $\mathbf{e}_K(\theta_t)$. To do this privately, the basis vector is treated by the user as a secret and a linear threshold secret-sharing scheme is used to generate shares that are sent to the servers. The servers return the inner products of the stored data and the secret-shared basis vector, which effectively form the secret shares of the desired inner product. Once the user collects sufficiently many secret shares, (s)he is able to retrieve the desired inner product, and therefore the desired symbol $W_{\theta_t}$. This  is the key to the private read-operation. Now during the write phase, the user  wishes to update the database to the new state: $\mathbf{W}'=[W_1,\cdots,W_{\theta_t-1},W_{\theta_t}+\Delta_t,W_{\theta_t+1}\cdots,W_K]$, which can be expressed as $\mathbf{W}' = \mathbf{W}+\Delta_t\mathbf{e}_K(\theta_t)$. This can be accomplished by sending the secret-shares of $\Delta_t\mathbf{e}_K(\theta_t)$ to the $N$ servers. The key observation here is the following: since the servers (those that were available during the read phase) have already received secret shares of $\mathbf{e}_K(\theta_t)$, the cost of sending the secret-shares of $\Delta_t\mathbf{e}_K(\theta_t)$ is significantly reduced. Essentially, it suffices to send secret shares of $\Delta_t$ which can be multiplied with the secret shares of $\mathbf{e}_K(\theta_t)$ to generate secret shares of $\Delta_t\mathbf{e}_K(\theta_t)$ at the servers. This is much more efficient because $\Delta_t\mathbf{e}_K(\theta_t)$ is a $K\times 1$ vector, while the dimension of $\Delta_t$ is $1$ (scalar), and $K\gg 1$. This is the intuition behind the synergistic gains from the joint design of private read and write operations that are exploited by ACSA-RW. Note that the servers operate directly on secret shares, as in homomorphic encryption \cite{Duong_Mishra_Yasuda}, and that these operations (inner products) are special cases of secure distributed matrix multiplications. Since CSA codes have been shown to be natural solutions for secure distributed matrix multiplications \cite{Jia_Jafar_CDBC}, it is intuitively to be expected that CSA schemes should lead to communication-efficient solutions for the private read-write implementation as described above.

\subsubsection{How to Fully Update a Distributed Database that is only Partially Accessible}\label{sec:howtowrite}
A seemingly paradoxical aspect of the write phase of ACSA-RW is that it is able to force the distributed database across all $N$ servers to be fully consistent with the updated data, even though the database is only partially accessible due to write-dropout servers. Let us explain the intuition behind this  with a toy example\footnote{The toy example is not strictly a special case of the ACSA-RW scheme, because the number of servers $N=2$ is too small to guarantee any privacy. However, the example serves to demonstrate the key idea.} where we have $N=2$ servers and $X=1$ security level is required. For simplicity we have only one file/submodel, i.e., $K=1$, so there are no privacy concerns. The storage at the two servers is $S_1=W+Z$, $S_2=Z$, where $W$ is the data (submodel) symbol, and $Z$ is the random noise symbol used to guarantee the $X=1$ security level, i.e., the storage at each server individually reveals nothing about the data $W$. The storage can be expressed in the following form.
          \begin{align}
              \begin{bmatrix}
                  S_1 \\
                  S_2
              \end{bmatrix}=\underbrace{\begin{bmatrix}
                      1 & 1 \\
                      0 & 1
                  \end{bmatrix}}_{\mathbf{G}}\begin{bmatrix}
                  W \\
                  Z
              \end{bmatrix}.
          \end{align}
          so that $\mathbf{G}$ is the coding function, and  the data can be recovered as $S_1-S_2=W$.
          
          Now suppose the data $W$ needs to be updated to the new value $W'=W+\Delta$. The updated storage $S_1', S_2'$ should be such that the coding function is unchanged (still the same $\mathbf{G}$ matrix) and the updated data can similarly be recovered as $S_1'-S_2'=W'$. Remarkably this can be done even if one of the two servers drops out and is therefore inaccessible. For example, if Server $1$ drops out, then we can update the storage only at Server $2$ to end up with $S_1' = S_1 = W+Z$, and $S_2' = Z-\Delta$, such that indeed $S_1'-S_2'=W'$ and the coding function is unchanged.
                 \begin{align}
             \begin{bmatrix}
                  S_1' \\
                  S_2'
              \end{bmatrix}=  \begin{bmatrix}
                  S_1 \\
                  S_2'
              \end{bmatrix}={\begin{bmatrix}
                      1 & 1 \\
                      0 & 1
                  \end{bmatrix}}\begin{bmatrix}
                  W+\Delta \\
                  Z-\Delta
              \end{bmatrix}={\begin{bmatrix}
                      1 & 1 \\
                      0 & 1
                  \end{bmatrix}}\begin{bmatrix}
                  W' \\
                  Z'
              \end{bmatrix}.
          \end{align}
          
 Similarly, if Server $2$ drops out, then we can update the storage only at Server $1$ to end up with $S_1' =  W+\Delta+Z$, and $S_2' = S_2=Z$, such that we still  have $S_1'-S_2'=W'$ and the coding function is unchanged.
          \begin{align}
             \begin{bmatrix}
                  S_1' \\
                  S_2'
              \end{bmatrix}=  \begin{bmatrix}
                  S_1' \\
                  S_2
              \end{bmatrix}={\begin{bmatrix}
                      1 & 1 \\
                      0 & 1
                  \end{bmatrix}}\begin{bmatrix}
                  W+\Delta \\
                  Z
              \end{bmatrix}
              ={\begin{bmatrix}
                      1 & 1 \\
                      0 & 1
                  \end{bmatrix}}\begin{bmatrix}
                  W' \\
                  Z
              \end{bmatrix}.
          \end{align}
 This, intuitively, is how the paradox is resolved, and a distributed database is fully updated even when it is only partially accessible. Note that the realization of the ``noise'' symbol is different for various realizations of the dropout servers, $Z'\neq Z$. 
 
While the toy example conveys a key idea, the generalization of this idea to the ACSA-RW scheme is rather non-trivial.  Let us shed some light on this generalization, which is made possible by the construction of an ``\emph{ACSA null-shaper}'', see Definition \ref{def:ACSAns}. Specifically, by carefully placing nulls of the CSA code polynomial in the update equation, the storage of the write-dropout servers in $\mathcal{S}_w^{(t)}$  is left unmodified. It is important to point out that the storage structure (i.e., ACSA storage, see Definition \ref{def:ACSAstor}) is preserved, just as the coding function is left unchanged in the toy example above. Let us demonstrate the idea with a minimal example where $X=2, T=1, X_{\Delta}=0, N=4$. Let us define the following functions.
          \begin{align}
              {\bf S}(\alpha) & =\mathbf{W}+\alpha\mathbf{Z}_1+\alpha^2\mathbf{Z}_2, \\
              {\bf Q}(\alpha) & =\mathbf{e}_{K}(\theta_t)+\alpha\mathbf{Z}',
          \end{align}
          where $\mathbf{W}=[W_1, W_2, \cdots, W_K]$ is a $K\times 1$ vector of the $K$ data (submodel) symbols, $\mathbf{Z}_1,\mathbf{Z}_2,\mathbf{Z}'$ are uniformly and independently distributed noise vectors that are used to protect data security and the user's privacy, respectively. Let $\alpha_1,\alpha_2,\alpha_3,\alpha_4$ be $4$ distinct non-zero elements from a finite field $\mathbb{F}_q$. The storage at the $4$ servers is ${\bf S}(\alpha_1)$, ${\bf S}(\alpha_2)$, ${\bf S}(\alpha_3)$, ${\bf S}(\alpha_4)$, respectively. Similarly, the read-queries for the $4$ servers are ${\bf Q}(\alpha_1)$, ${\bf Q}(\alpha_2)$, ${\bf Q}(\alpha_3)$, ${\bf Q}(\alpha_4)$, respectively. We note that the storage vectors and the query vectors can be viewed as secret sharings of ${\bf W}$ and ${\bf e}_K(\theta_t)$ vectors with threshold of $X=2$ and $T=1$, respectively. Now let us assume that $\mathcal{S}_w^{(t)}=\{1\}$ for some $t\in\mathbb{N}$, i.e., Server 1 drops out. Let us define the function $\Omega(\alpha)=(\alpha_1-\alpha)/\alpha_1$, which is referred to as {\it ACSA null-shaper}, and consider the following update equation.
          \begin{align}
              {\bf S}'(\alpha)={\bf S}(\alpha)+\Omega(\alpha)\Delta_t {\bf Q}(\alpha).
          \end{align}
          Inspecting the second term on the RHS, we note that
          \begin{align}
              \Omega(\alpha)\Delta_t {\bf Q}(\alpha) & =\frac{1}{\alpha_1}(\alpha_1-\alpha)\Delta_t(\mathbf{e}_{K}(\theta_t)+\alpha\mathbf{Z}')                                                          \\
                                               & =\frac{1}{\alpha_1}\Delta_t\left(\alpha_1\mathbf{e}_{K}(\theta_t)+\alpha(\alpha_1\mathbf{Z}'-\mathbf{e}_{K}(\theta_t))-\alpha^2\mathbf{Z}'\right) \\
                                               & =\Delta_t\left(\mathbf{e}_{K}(\theta_t)+\alpha(\mathbf{Z}'-\alpha_1^{-1}\mathbf{e}_{K}(\theta_t))-\alpha^2\alpha_1^{-1}\mathbf{Z}'\right),        \\
                                               & =\Delta_t\left(\mathbf{e}_{K}(\theta_t)+\alpha\dot{\mathbf{I}}_1-\alpha^2\dot{\mathbf{I}}_2\right),
          \end{align}
          where $\dot{\mathbf{I}}_1=\mathbf{Z}'-\alpha_1^{-1}\mathbf{e}_{K}(\theta_t)$ and $\dot{\mathbf{I}_2}=\alpha_1^{-1}\mathbf{Z}'$. Evidently, by the update equation, the user is able to update the symbol of the $\theta_t^{th}$ message with the increment $\Delta_t$, while maintaining the storage as a secret sharing of threshold $2$, i.e.,
          \begin{align}
              {\bf S'}(\alpha)=\left(\mathbf{W}+\Delta_t\mathbf{e}_{k}(\theta_t)\right)+\alpha(\mathbf{Z}_1+\Delta_t\dot{\mathbf{I}}_1)+\alpha^2(\mathbf{Z}_2+\Delta_t\dot{\mathbf{I}}_2).
          \end{align}
          However, by the definition of ACSA null-shaper, we have $\Omega(\alpha_1)=0$. Thus ${\bf S'}(\alpha_1)={\bf S}(\alpha_1)$, and we do not have to update the storage at the Server 1. In other words, $\dot{\mathbf{I}}_1$ and $\dot{\mathbf{I}}_2$ are {\it artificially correlated interference symbols} such that the codeword $\Omega(\alpha_1){\bf Q}(\alpha_1)$ is zero, and accordingly, the storage of Server 1 is left unmodified. Note that ACSA null-shaper does not affect the storage structure because $X=2>T=1$. The idea illustrated in this minimal example indeed generalizes to the full ACSA-RW scheme, see Section \ref{sec:acsarwproof} for details.

\subsubsection{Comparison with \cite{Kim_Lee_FSL}}\label{sec:kim}
Let us compare our ACSA-RW solution with that in \cite{Kim_Lee_FSL}. The setting in \cite{Kim_Lee_FSL} corresponds to $X=0, T=1, X_{\Delta}=0$, and $|\mathcal{S}_w^{(t)}|=|\mathcal{S}_r^{(t)}|=0$ for all $t\in\mathbb{N}$. Note that our ACSA-RW scheme for $X=1, T=1, X_{\Delta}=0$ and $K_c=1$ applies to the setting of \cite{Kim_Lee_FSL}  ($X=1$ security automatically satisfies $X=0$ security). To make the comparison more transparent, let us briefly review the construction in \cite{Kim_Lee_FSL}, where at any time $t, t\in\mathbb{Z}^*$, each of the $N$ servers stores the $K$ submodels in the following coded form.
          \begin{align}
              \mathbf{W}_{k}^{(t-1)}+z_{k}^{(t)}\mathbf{\Delta}_{t}, \forall k\in[K].
          \end{align}
         For all $t\in\mathbb{N}$, $\left(z_{k}^{(t)}\right)_{k\in[K]}$ are distinct random scalars generated by User $t$ and we set $z_{\theta_t}^{(t)}=1$. For completeness we define $\mathbf{\Delta}_0=\mathbf{0}, z_{k}^{(0)}=0, \mathbf{W}_k^{(-1)}=\mathbf{W}_k^{(0)}, \forall k\in[K]$. In addition, the $N$ servers store the random scalars $\left(z_{k}^{(t)}\right)_{k\in[K]}$, as well as the increment $\mathbf{\Delta}_t$ according to a secret sharing scheme of threshold $1$. In the retrieval phase, User $t$ retrieves the coded desired submodel $\mathbf{W}_{\theta_t}^{(t-2)}+z_{\theta_t}^{(t-1)}\mathbf{\Delta}_{t-1}$ privately according to a capacity-achieving replicated storage based PIR scheme, e.g., \cite{PIRfirstjournal}. Besides, User $t$ also downloads the secret shared random scalars $\left(z_{k}^{(t-1)}\right)_{k\in[K]}$ and the increment $\mathbf{\Delta}_{t-1}$ to correctly recover the desired submodel $\mathbf{W}_{\theta_t}^{(t-1)}$. In the update phase, User $t$ uploads to each of the $N$ servers the following update vectors $\mathbf{P}_k^{(t)}$ for all $k\in[K]$.
          \begin{align}
              \mathbf{P}_k^{(t)}=\left\{
              \begin{array}{ll}
                  z_{k}^{(t)}\mathbf{\Delta}_t-z_{k}^{(t-1)}\mathbf{\Delta}_{t-1}, & k\neq \theta_{t-1}, \\
                  z_{k}^{(t)}\mathbf{\Delta}_t,                                    & k=\theta_{t-1}.
              \end{array}
              \right.\label{eq:kimupdate}
          \end{align}
          Also, User $t$ uploads the secret shared random scalars $\left(z_{k}^{(t)}\right)_{k\in[K]}$ and the increment $\mathbf{\Delta}_{t}$ to the $N$ servers. To perform an update, each of the $N$ servers updates all of the submodels $k\in[K]$ according to the following equation.
          \begin{align}
              \left(\mathbf{W}_{k}^{(t-2)}+z_{k}^{(t-1)}\mathbf{\Delta}_{t-1}\right)+\mathbf{P}_k^{(t)}=\mathbf{W}_{k}^{(t-1)}+z_{k}^{(t)}\mathbf{\Delta}_{t}.
          \end{align}
          Perhaps the most significant difference between our ACSA-RW scheme and the construction in \cite{Kim_Lee_FSL} is that the latter does not guarantee the privacy of successive updates, i.e., by monitoring the storage at multiple time slots, the servers are eventually able to learn about the submodel indices from past updates\footnote{This is because the update vectors \eqref{eq:kimupdate} for the $K$ submodels at any time $t$ can be viewed as $\mathcal{P}_t=\vspan\{\mathbf{\Delta}_t,\mathbf{\Delta}_{t-1}\}$. For any two consecutive time slots $t$ and $t+1$, it is possible to determine $\vspan\{\mathbf{\Delta}_t\}=\vspan\{\mathbf{\Delta}_t,\mathbf{\Delta}_{t-1}\}\cap\vspan\{\mathbf{\Delta}_{t+1},\mathbf{\Delta}_{t}\}=\mathcal{P}_t\cap\mathcal{P}_{t+1}$. Due to the fact that for User $t$, the update vector for the $\theta_{t-1}^{th}$ submodel only lies in $\vspan\{\mathbf{\Delta}_t\}$, any curious server is able to obtain information about $\theta_{t-1}$ from $\mathcal{P}_{t}$ if $\mathbf{\Delta}_t$ is linearly independent of $\mathbf{\Delta}_{t-1}$.}.  On the other hand, our construction guarantees  information theoretic privacy  for an unlimited number of updates,  without extra storage overhead. Furthermore, we note that the normalized download cost achieved by the construction in \cite{Kim_Lee_FSL} cannot be less than $2$, whereas ACSA-RW achieves download cost of less than $2$ with large enough $N$. For the asymptotic setting $L/K\rightarrow\infty$, the upload cost\footnote{In \cite{Kim_Lee_FSL}  the achieved upload cost is $NLK+L+K$. However, in terms of average upload compression, e.g., by entropy encoding, allows lower upload cost. For example, in the asymptotic setting $L/K\rightarrow\infty$, the upload cost of $(2N+1+1/(N-1))$ may be achievable.} achieved by \cite{Kim_Lee_FSL} is at least $2N+1$, while ACSA-RW achieves the upload cost of at most $N$. The lower bound of upload cost of XSTPFSL is characterized in \cite{Kim_Lee_FSL} as $NK$. However, our construction of ACSA-RW shows that it is possible to do better.\footnote{It is assumed in \cite{Kim_Lee_FSL} that for the desired submodel, the user uploads $L$ symbols for the update, while for other submodels, the user should also upload $(K-1)L$ symbols to guarantee the privacy. However, it turns out that the uploaded symbols for the update of the desired submodel and the symbols for the purpose of guaranteeing the privacy do not have to be independent.} In particular, for the asymptotic setting $K\rightarrow\infty, L/K\rightarrow\infty$, the upload cost of less than $N$ is achievable by the ACSA-RW scheme.

\subsubsection{Comparison with \cite{Kushilevitz_DORAM_Sub}}
Let us also briefly review the $4$-server information-theoretic DORAM construction in \cite{Kushilevitz_DORAM_Sub} to see how our ACSA-RW scheme improves upon it. Note that the setting considered in \cite{Kushilevitz_DORAM_Sub} is a special case of our problem where $X=1,T=1,X_{\Delta}=0, K_c=1$ and $|\mathcal{S}_w^{(t)}|=|\mathcal{S}_r^{(t)}|=0$ for all $t\in\mathbb{N}$. First, we note that in the asymptotic setting $L/K\rightarrow\infty$, the upload cost achieved by the ACSA-RW is the same as that in \cite{Kushilevitz_DORAM_Sub}. Therefore, for this comparison we focus on the read phase and the download cost. Specifically, the information-theoretic DORAM construction in \cite{Kushilevitz_DORAM_Sub} partitions the four servers into two groups, each of which consists of $2$ servers. For the retrieval phase, the first group emulates a $2$-server PIR, storing the $K$ submodels secured with additive random noise, i.e., $\mathbf{W}+\mathbf{Z}$. The second group emulates another  $2$-server PIR storing the random noise $\mathbf{Z}$. To retrieve the desired submodel privately, the user exploits a PIR scheme to retrieve the desired secured submodel, as well as the corresponding random noise. Therefore, with capacity-achieving PIR schemes, the download cost is $4$ (for large $K$). On the other hand, our ACSA-RW scheme avoids the partitioning of the servers and improves the download cost by jointly exploiting all  $4$ servers. Remarkably, with the idea of cross-subspace alignment, out of the $4$ downloaded symbols, the interference symbols align within $2$ dimensions, leaving $2$ dimensions interference-free for the desired symbols, and consequently, the asymptotically optimal download cost of $2$ is achievable. Lastly, the ACSA-RW scheme also generalizes efficiently to arbitrary numbers of servers.

  \subsubsection{On the Assumption $L\gg K\gg 1$ for FSL}\label{sec:numerical}
   Finally, let us briefly explore the practical relevance of the asymptotic limits $K\rightarrow\infty, L/K\rightarrow\infty$, with an example. Suppose we have $N=6$ distributed servers,  we require security and privacy levels of $X_{\Delta}=T=1, X=3$. Let us set $K_c=1$, and we operate over $\mathbb{F}_8$. Consider an e-commerce recommendation application similar to what is studied in \cite{SFSM}, where a global model with a total of $3,500,000$ symbols (from $\mathbb{F}_8$) is partitioned into $K=50$ submodels. Each of the submodels is comprised of $L=70,000$ symbols. Note that $L\gg K\gg 1$. Let us assume that there are no dropout servers for some $t\in\mathbb{N}$. Now according to Lemma \ref{lemma:write},  the normalized upload cost achieved by the ACSA-RW scheme is $U_t=\frac{6\times 35000+6\times2\times50}{70,000}\approx 3.00857$. On the other hand, the normalized download cost achieved is $D_t=6/2=3$. Evidently, the asymptotic limits $K\rightarrow\infty, L/K\rightarrow\infty$ are fairly accurate for this non-asymptotic setting. For this particular example, the upload cost is increased by only $0.29\%$ compared to the asymptotic limit. On the other hand, the download cost is increased by only $1.4\times 10^{-22}\%$ compared to the lower bound (evaluated for $K=50$) from  \cite{Jia_Sun_Jafar_XSTPIR}.

\section{Proof of Theorem \ref{thm:acsarw}}\label{sec:acsarwproof}
For all $t\in\mathbb{N}$, we require that the number of read and write dropout servers is less than the corresponding threshold values, $0\leq |\mathcal{S}_r^{(t)}|<S_r^{\text{\normalfont thresh}}$ and $0\leq |\mathcal{S}_w^{(t)}|<S_w^{\text{\normalfont thresh}}$. Since the read and write dropout thresholds cannot be less than $1$, we require that $X\geq X_{\Delta}+T$, and  $N\geq K_c+X+T$ for a positive integer $K_c$.  Let us define $J=\xi\cdot\lcm\left([S_r^{\text{\normalfont thresh}}]\cup[S_w^{\text{\normalfont thresh}}]\right)$ and we set $L=J K_c$, where $\xi$ is a positive\footnote{The purpose of $\xi$ is primarily to allow the scheme to scale to larger values of $L$, one could assume $\xi=1$ for simplicity.} integer. In other words, $i\mid J$ for all $i\in[S_r^{\text{\normalfont thresh}}]\cup[S_w^{\text{\normalfont thresh}}]$. For ease of reference, let us define $R_r^{(t)}\triangleq S_r^{\text{\normalfont thresh}}-|\mathcal{S}_r^{(t)}|$, $\#_r^{(t)}\triangleq J/R_r^{(t)}$, $R_w^{(t)}\triangleq S_w^{\text{\normalfont thresh}}-|\mathcal{S}_w^{(t)}|$ and $\#_w^{(t)}\triangleq J/R_w^{(t)}$ for all $t\in\mathbb{N}$. Note that it is guaranteed by the choice of $J$ that $\#_r^{(t)}, \#_w^{(t)}$ are positive integers for all $t\in\mathbb{N}$. Indeed in the ACSA-RW scheme, private read and write operations can be viewed as operations that consist of $K_c\#_r^{(t)}$ and $K_c\#_w^{(t)}$ sub-operations, and in each sub-operation, $R_r^{(t)}$ and $R_w^{(t)}$ symbols of the desired submodel are retrieved and updated, respectively. The choice of $J$ guarantees that the number of sub-operations is always an integer, regardless of $|\mathcal{S}_r^{(t)}|$ and $|\mathcal{S}_w^{(t)}|$. The parameter $\xi$ guarantees that $L$ is still a free parameter so that $L/K\rightarrow\infty$ is well-defined. In other words, $L$ can be any multiple of $K_c\cdot\lcm\left([S_r^{\text{\normalfont thresh}}]\cup[S_w^{\text{\normalfont thresh}}]\right)$. Let us define $\mu\triangleq \max(S_r^{\text{\normalfont thresh}},S_w^{\text{\normalfont thresh}})$. We will need a total of $N+\max(\mu,K_c)$ distinct elements from the finite field $\mathbb{F}_q, q\geq N+\max(\mu,K_c)$, denoted as $(\alpha_1,\alpha_2,\cdots,\alpha_N)$, $(\widetilde{f}_1,\widetilde{f}_2,\cdots,\widetilde{f}_{\max(\mu,K_c)})$. Let us define the set $\overline{\mathcal{S}}_w^{(t)}=[N]\setminus\mathcal{S}_w^{(t)}$, and the set $\overline{\mathcal{S}}_r^{(t)}=[N]\setminus\mathcal{S}_r^{(t)}$. For all $t\in\mathbb{Z}^*, j\in[J], k\in[K], i\in[K_c]$, let us define
\begin{align}
    W_{k}^{(t)}(j,i)=W_{k}^{(t)}(i+K_c(j-1)),
\end{align}
i.e., the $L=J K_c$ symbols of each of the $K$ messages are reshaped into a $J\times K_c$ matrix. Similarly, for all $t\in\mathbb{N}, j\in[J], i\in[K_c]$, we define
\begin{align}
    \Delta^{(t)}_{j,i}=\Delta_t(i+K_c(j-1)).
\end{align}

For all $t\in\mathbb{Z}^*, j\in[J], i\in[K_c]$, let us define the following vectors.
\begin{align}
    \dot{\mathbf{W}}_{j,i}^{(t)}=\left[W_{1}^{(t)}(j,i),W_{2}^{(t)}(j,i),\cdots,W_{K}^{(t)}(j,i)\right]^{\mathsf{T}}.
\end{align}
Further, let us set $\mathcal{Z}_s=\left\{\dot{\mathbf{Z}}_{j,x}^{(0)}\right\}_{j\in[J],x\in[X]}$, where $\dot{\mathbf{Z}}_{j,x}^{(0)}$ are i.i.d. uniform column vectors from $\mathbb{F}_q^K$, $\forall j\in[J],x\in[X]$. For all $t\in\mathbb{N}, j\in[J],x\in[X]$, let $\dot{\mathbf{Z}}_{j,x}^{(t)}$ be $K\times 1$ column vectors from the finite field $\mathbb{F}_q$. For all $t\in\mathbb{N},u\in[\mu],i\in[K_c],s\in[T]$, let $\widetilde{\mathbf{Z}}_{u,i,s}^{(t)}$ be i.i.d. uniform column vectors from $\mathbb{F}_q^K$, and let us set $\mathcal{Z}_U^{(t)}=\left\{\widetilde{\mathbf{Z}}_{u,i,s}^{(t)}\right\}_{u\in[\mu],i\in[K_c],s\in[T]}\cup\left\{\dddot{Z}_{\ell,i,x}^{(t)}\right\}_{\ell\in[\#_w^{(t)}],i\in[K_c],x\in[X_{\Delta}]}$, where for all $t\in\mathbb{N}, \ell\in[\#_w^{(t)}],i\in[K_c],x\in[X_{\Delta}]$, $\dddot{Z}_{\ell,i,x}^{(t)}$ are i.i.d. uniform scalars from the finite field $\mathbb{F}_q$.

The ACSA-RW scheme (Definition \ref{def:acsarw}) is built upon the elements introduced in Definitions \ref{def:polasn}--\ref{def:ACSAns}.

\begin{definition}\label{def:polasn}{\bf(Pole Assignment) }
    If $\mu\geq K_c$, let us define the $\mu\times K_c$ matrix $\mathbf{F}$ to be the first $K_c$ columns of the following $\mu\times\mu$ matrix.
    \begin{align}
        \begin{bmatrix}
            \widetilde{f}_{1}     & \widetilde{f}_{\mu}   & \dots               & \widetilde{f}_{3} & \widetilde{f}_{2}   \\
            \widetilde{f}_{2}     & \widetilde{f}_{1}     & \widetilde{f}_{\mu} &                   & \widetilde{f}_{3}   \\
            \vdots                & \widetilde{f}_{2}     & \widetilde{f}_{1}   & \ddots            & \vdots              \\
            \widetilde{f}_{\mu-1} &                       & \ddots              & \ddots            & \widetilde{f}_{\mu} \\
            \widetilde{f}_{\mu}   & \widetilde{f}_{\mu-1} & \dots               & \widetilde{f}_{2} & \widetilde{f}_{1}
        \end{bmatrix}.
    \end{align}
    On the other hand, if $\mu< K_c$, let us define the $\mu\times K_c$ matrix $\mathbf{F}$ to be the first $\mu$ rows of the following $K_c\times K_c$ matrix.
    \begin{align}
        \begin{bmatrix}
            \widetilde{f}_{1}   & \widetilde{f}_{2}   & \dots             & \widetilde{f}_{K_c-1} & \widetilde{f}_{K_c}   \\
            \widetilde{f}_{K_c} & \widetilde{f}_{1}   & \widetilde{f}_{2} &                       & \widetilde{f}_{K_c-1} \\
            \vdots              & \widetilde{f}_{K_c} & \widetilde{f}_{1} & \ddots                & \vdots                \\
            \widetilde{f}_{3}   &                     & \ddots            & \ddots                & \widetilde{f}_{2}     \\
            \widetilde{f}_{2}   & \widetilde{f}_{3}   & \dots             & \widetilde{f}_{K_c}   & \widetilde{f}_{1}
        \end{bmatrix}.
    \end{align}
    Let $\left(f_{j,i}\right)_{j\in[J],i\in[K_c]}$ be a total of $JK_c$ elements from the finite field $\mathbb{F}_q$, which are defined as follows.
    \begin{align}
        \begin{bmatrix}
            f_{1,1} & f_{1,2} & \cdots & f_{1,K_c} \\
            f_{2,1} & f_{2,2} & \cdots & f_{2,K_c} \\
            \vdots  & \vdots  & \vdots & \vdots    \\
            f_{J,1} & f_{J,2} & \cdots & f_{J,K_c} \\
        \end{bmatrix}=
        \begin{bmatrix}
            \mathbf{F} \\
            \mathbf{F} \\
            \vdots     \\
            \mathbf{F}
        \end{bmatrix}.
    \end{align}
\end{definition}
According to the definition, we have the following two propositions immediately.
\begin{proposition}\label{prop:distinj}
    For all $i\in[K_c],m,n\in[J]$ such that $m\leq n$, $|[m:n]|\leq \mu$, the constants $\left(f_{j,i}\right)_{j\in[m:n]}$ are distinct.
\end{proposition}
\begin{proposition}\label{prop:distink}
    For all $j\in[J]$, the constants $\left(f_{j,i}\right)_{i\in[K_c]}$ are distinct.
\end{proposition}
\begin{definition}\label{def:noiasn}{\bf(Noise Assignment) }
    For all $t\in\mathbb{N}, i\in[K_c], s\in[T]$, let us define
    \begin{align}
         & \left(\ddot{\mathbf{Z}}_{1,i,s}^{(t)},\ddot{\mathbf{Z}}_{2,i,s}^{(t)},\cdots,\ddot{\mathbf{Z}}_{J,i,s}^{(t)}\right)\notag                                                                                                                                                                                                                                                                     \\
         & =\left(\widetilde{\mathbf{Z}}_{1,i,s}^{(t)},\widetilde{\mathbf{Z}}_{2,i,s}^{(t)},\cdots,\widetilde{\mathbf{Z}}_{\mu,i,s}^{(t)},\widetilde{\mathbf{Z}}_{1,i,s}^{(t)},\widetilde{\mathbf{Z}}_{2,i,s}^{(t)},\cdots,\widetilde{\mathbf{Z}}_{\mu,i,s}^{(t)},\cdots,\widetilde{\mathbf{Z}}_{1,i,s}^{(t)},\widetilde{\mathbf{Z}}_{2,i,s}^{(t)},\cdots,\widetilde{\mathbf{Z}}_{\mu,i,s}^{(t)}\right),
    \end{align}
    i.e., $\left(\widetilde{\mathbf{Z}}_{1,i,s}^{(t)},\widetilde{\mathbf{Z}}_{2,i,s}^{(t)},\cdots,\widetilde{\mathbf{Z}}_{\mu,i,s}^{(t)}\right)$ are assigned in cyclic order to $\left(\ddot{\mathbf{Z}}_{1,i,s}^{(t)},\ddot{\mathbf{Z}}_{2,i,s}^{(t)},\cdots,\ddot{\mathbf{Z}}_{J,i,s}^{(t)}\right)$.
\end{definition}
\begin{definition}\label{def:ACSAstor}{\bf(ACSA Storage) }
    For any $t\in\mathbb{Z}^*$, the storage at the $N$ servers is said to form an ACSA storage if for all $n\in[N]$, $\mathbf{S}_{n}^{(t)}$ has the following form.
    \begin{align}
        \mathbf{S}_n^{(t)}=
        \begin{bmatrix}
            \sum_{i\in[K_c]}\frac{1}{\alpha_n-f_{1,i}}\dot{\mathbf{W}}^{(t)}_{1,i}+\sum_{x\in[X]}\alpha_n^{x-1}\dot{\mathbf{Z}}_{1,x}^{(t)} \\
            \sum_{i\in[K_c]}\frac{1}{\alpha_n-f_{2,i}}\dot{\mathbf{W}}^{(t)}_{2,i}+\sum_{x\in[X]}\alpha_n^{x-1}\dot{\mathbf{Z}}_{2,x}^{(t)} \\
            \vdots                                                                                                                          \\
            \sum_{i\in[K_c]}\frac{1}{\alpha_n-f_{J,i}}\dot{\mathbf{W}}^{(t)}_{J,i}+\sum_{x\in[X]}\alpha_n^{x-1}\dot{\mathbf{Z}}_{J,x}^{(t)}
        \end{bmatrix}.
    \end{align}
    Note that $X$ i.i.d. uniform random noise terms are MDS coded to guarantee the $X$-security.
\end{definition}

\begin{definition}\label{def:ACSAquery}{\bf(ACSA Query) }For any $t\in\mathbb{N}$, the read-queries $\left(Q_{n}^{(t,\theta_t)}\right)_{n\in[N]}$ by User $t$ for the $N$ servers are said to form an ACSA query if for all $n\in[N]$, we have
    \begin{align}
        Q_{n}^{(t,\theta_t)}=\left(\mathbf{Q}^{(t,\theta_t)}_{n,1},\mathbf{Q}^{(t,\theta_t)}_{n,2},\cdots,\mathbf{Q}^{(t,\theta_t)}_{n,K_c}\right),
    \end{align}
    where for all $i\in[K_c]$
    \begin{align}
        \mathbf{Q}^{(t,\theta_t)}_{n,i}=\begin{bmatrix}
            \mathbf{e}_K(\theta_t)+(\alpha_n-f_{1,i})\sum_{s\in[T]}\alpha_n^{s-1}\ddot{\mathbf{Z}}_{1,i,s}^{(t)} \\
            \mathbf{e}_K(\theta_t)+(\alpha_n-f_{2,i})\sum_{s\in[T]}\alpha_n^{s-1}\ddot{\mathbf{Z}}_{2,i,s}^{(t)} \\
            \vdots                                                                                               \\
            \mathbf{e}_K(\theta_t)+(\alpha_n-f_{J,i})\sum_{s\in[T]}\alpha_n^{s-1}\ddot{\mathbf{Z}}_{J,i,s}^{(t)}
        \end{bmatrix}.
    \end{align}
    Note that for all $t\in\mathbb{N}, n\in[N]$, $H\left(Q_n^{(t,\theta_t)}\right)=\mu K_c K$ in $q$-ary units, because according to Definition \ref{def:polasn} and Definition \ref{def:noiasn}, for all $i\in[K_c]$, $\mathbf{Q}_{n,i}^{(t,\theta_t)}$ is uniquely determined by its first $\mu K$ entries (the rest are replicas). Also note that $T$ i.i.d. uniform random noise terms are MDS coded to guarantee the $T$-privacy.
\end{definition}
\begin{remark}
Indeed, the fact that $\mathbf{Q}_{n,i}^{(t,\theta_t)}$ is uniquely determined by its first $\mu K$ entries is jointly guaranteed by the cyclic structures of pole assignment and noise assignment, i.e., Definition \ref{def:polasn} and Definition \ref{def:noiasn}. These cyclic structures are important in terms of minimizing the entropy of ACSA queries, so that it does not scale with the number of symbols $L$.
\end{remark}

\begin{definition}\label{def:ACSAinc}{\bf(ACSA Increment) }For any $t\in\mathbb{N}$, the write-queries $\left(P_{n}^{(t,\theta_t)}\right)_{n\in[N]}$ by User $t$ for the $N$ servers are said to form an ACSA increment if for all $n\in[N]$, we have
    \begin{align}
        P_{n}^{(t,\theta_t)}=\left(\mathbf{P}_{n,1}^{(t,\theta_t)},\mathbf{P}_{n,2}^{(t,\theta_t)},\cdots,\mathbf{P}_{n,K_c}^{(t,\theta_t)}\right),
    \end{align}
    where for all $i\in[K_c]$,
    \begin{align}
        \mathbf{P}_{n,i}^{(t,\theta_t)}=\diag\left(\widetilde{\Delta}_{1,i}^{(t)}\mathbf{I}_{KR_w^{(t)}},\widetilde{\Delta}_{2,i}^{(t)}\mathbf{I}_{KR_w^{(t)}},\cdots,\widetilde{\Delta}_{\#_w^{(t)},i}^{(t)}\mathbf{I}_{KR_w^{(t)}}\right),
    \end{align}
    and for all $\ell\in[\#_w^{(t)}]$,
    \begin{align}
        \widetilde{\Delta}_{\ell,i}^{(t)}=\sum_{j\in[(\ell-1)R_w^{(t)}+1:\ell R_w^{(t)}]}\frac{1}{\alpha_n-f_{j,i}}\Delta_{j,i}^{(t)}+\sum_{x\in[X_{\Delta}]}\alpha_n^{x-1}\dddot{Z}_{\ell,i,x}^{(t)}.
    \end{align}
    Note that for all $t\in\mathbb{N}, n\in[N]$, $H\left(P_n^{(t,\theta_t)}\right)=\#_w^{(t)} K_c$ in $q$-ary units. This is because for all $i\in[K_c]$, $\mathbf{P}_{n,i}^{(t,\theta_t)}$ is uniquely determined by $\left(\widetilde{\Delta}_{\ell,i}^{(t)}\right)_{\ell\in[\#_w^{(t)}]}$. Also note that $X_{\Delta}$ MDS coded i.i.d. uniform noise terms are used to guarantee the  $X_{\Delta}$-security.
\end{definition}

\begin{definition}{\bf(ACSA Packer) }For all $t\in\mathbb{N}, n\in[N]$, the ACSA Packer is defined as follows.
    \begin{align}
        \Xi_n^{(t)}=\left(\mathbf{\Xi}_{n,\ell,i}^{(t)}\right)_{\ell\in[\#_r^{(t)}],i\in[K_c]},
    \end{align}
    where for all $\ell\in[\#_r^{(t)}]$, $i\in[K_c]$,
    \begin{align}
        \mathbf{\Xi}_{n,\ell,i}^{(t)} & =\diag\left(\frac{\prod_{i'\in[K_c]\setminus\{i\}}(\alpha_n-f_{1,i'})}{\prod_{i'\in[K_c]\setminus\{i\}}(f_{1,i}-f_{1,i'})}\mathbf{I}_K,\cdots,\frac{\prod_{i'\in[K_c]\setminus\{i\}}(\alpha_n-f_{J,i'})}{\prod_{i'\in[K_c]\setminus\{i\}}(f_{J,i}-f_{J,i'})}\mathbf{I}_K\right)\notag   \\
                                   & \quad\times \diag(\underbrace{0\mathbf{I}_K,\cdots,0\mathbf{I}_K}_{\text{$R_r^{(t)}(\ell-1)$ $0\mathbf{I}_K$'s}},\underbrace{1\mathbf{I}_K,\cdots,1\mathbf{I}_K}_{\text{$R_r^{(t)}$ $1\mathbf{I}_K$'s}},\underbrace{0\mathbf{I}_K,\cdots,0\mathbf{I}_K}_{\text{$(J-\ell R_r^{(t)})$ $0\mathbf{I}_K$'s}}).
    \end{align}
    Note that the ACSA packer is a {\bf constant} since $|\mathcal{S}_r^{(t)}|$ is globally known.
\end{definition}

\begin{definition}{\bf (ACSA Unpacker) }For all $t\in\mathbb{N}, n\in[N]$, the ACSA Unpacker is defined as follows.
    \begin{align}
        \Upsilon_n^{(t)}=\left(\mathbf{\Upsilon}_{n,1}^{(t)},\mathbf{\Upsilon}_{n,2}^{(t)},\cdots,\mathbf{\Upsilon}_{n,K_c}^{(t)}\right),
    \end{align}
    where for all $i\in[K_c]$,
    \begin{align}
        \mathbf{\Upsilon}_{n,i}^{(t)}=\diag\left(\left(\frac{\prod_{j\in\mathcal{F}_{1}^{(t)}}(\alpha_n-f_{j,i})}{\prod_{j\in\mathcal{F}_{1}^{(t)}}(f_{1,i}-f_{j,i})}\right)\mathbf{I}_K,\cdots,\left(\frac{\prod_{j\in\mathcal{F}_{J}^{(t)}}(\alpha_n-f_{j,i})}{\prod_{j\in\mathcal{F}_{J}^{(t)}}(f_{J,i}-f_{j,i})}\right)\mathbf{I}_K\right),
    \end{align}
    where for all $j\in[J]$, we define $\mathcal{F}_{j}=\left[\left(\ceil*{j/R_w^{(t)}}-1\right)R_w^{(t)}+1:\ceil*{j/R_w^{(t)}}R_w^{(t)}\right]\setminus\{j\}$. Note that the ACSA unpacker is a {\bf constant} since $|\mathcal{S}_w^{(t)}|$ is globally known.
\end{definition}

\begin{remark}
As noted at the beginning of the proof, private read and write operations consists of $K_c\#_r^{(t)}$ and $K_c\#_w^{(t)}$ sub-operations, and for each sub-operation, $R_r^{(t)}$ and $R_w^{(t)}$ symbols of the desired submodel are retrieved and updated respectively. This is respectively made possible by the constructions of ACSA Packer and Unpacker. For private read, by ACSA Packer, in each sub-operation $R_r^{(t)}$ symbols of the desired submodel are ``packed'' into a Cauchy-Vandermonde structured answer string for best download efficiency, where Cauchy terms carry desired symbols and interference symbols are aligned along Vandermonde terms. Recoverability of the desired symbols is guaranteed by the invertibility of Cauchy-Vandermonde matrices (see, e.g., \cite{Gasca_Martinez_Muhlbach}). On the other hand, ACSA Increment (Definition \ref{def:ACSAinc}) can be viewed as a bunch of Cauchy-Vandermonde structured codewords, and in each codeword, a total of $R_w^{(t)}$ symbols are ``packed'' together for best upload efficiency. To correctly perform private write, each of the codewords is ``unpacked'' by the ACSA Unpacker to produce $R_w^{(t)}$ codewords that individually carry different increment symbols to preserve the ACSA Storage structure (Definition \ref{def:ACSAstor}).
\end{remark}

\begin{definition}{\bf (ACSA Null-shaper)}\label{def:ACSAns}
    For all $t\in\mathbb{N}, n\in[N]$, the ACSA null-shaper is defined as follows.
    \begin{align}
        \Omega^{(t)}=\left(\mathbf{\Omega}_{n,1}^{(t)},\mathbf{\Omega}_{n,2}^{(t)},\cdots,\mathbf{\Omega}_{n,K_c}^{(t)}\right),
    \end{align}
    where for all $i\in[K_c]$,
    \begin{align}
        \mathbf{\Omega}_{n,i}^{(t)}=\diag\left(\left(\frac{\prod_{m\in\mathcal{S}_w^{(t)}}(\alpha_n-\alpha_m)}{\prod_{m\in\mathcal{S}_w^{(t)}}(f_{1,i}-\alpha_m)}\right)\mathbf{I}_K,\cdots,\left(\frac{\prod_{m\in\mathcal{S}_w^{(t)}}(\alpha_n-\alpha_m)}{\prod_{m\in\mathcal{S}_w^{(t)}}(f_{J,i}-\alpha_m)}\right)\mathbf{I}_K\right).
    \end{align}
    Note that for all $n\in\mathcal{S}_w^{(t)}$, we have $\mathbf{\Omega}_{n}=\mathbf{0}$. Besides, the ACSA null-shaper is a {\bf constant} since $\mathcal{S}_w^{(t)}$ is globally known.
\end{definition}

\begin{definition}\label{def:acsarw}{\bf(ACSA-RW Scheme) }The initial storage at the $N$ servers is the ACSA storage at time $0$, i.e., $\left(\mathbf{S}_{n}^{(0)}\right)_{n\in[N]}$. At time $t, t\in\mathbb{N}$, in the read phase, the user uploads the ACSA query for the $N$ servers and retrieves the desired submodel $\mathbf{W}^{(t-1)}_{\theta_t}$ from the answers returned by the servers $n, n\in\overline{\mathcal{S}}_r^{(t)}$, which are constructed as follows.
    \begin{align}
        A_n^{(t,\theta_t)}=\left(\left(\mathbf{S}_n^{(t-1)}\right)^{\mathsf{T}}\mathbf{\Xi}_{n,\ell,i}^{(t)}\mathbf{Q}_{n,i}^{(t,\theta_t)}\right)_{\ell\in[\#_r^{(t)}],i\in[K_c]}, n\in\overline{\mathcal{S}}_r^{(t)}.\label{eq:answer}
    \end{align}
    In the update phase, the user uploads the ACSA increment for the servers $n, n\in\overline{\mathcal{S}}_w^{(t)}$, and each of the servers updates its storage according to the following equation.
    \begin{align}
        \mathbf{S}_n^{(t)} & =\mathbf{S}_n^{(t-1)}+\sum_{i\in[K_c]}\mathbf{\Omega}_{n,i}^{(t)}\mathbf{\Upsilon}_{n,i}^{(t)}\mathbf{P}_{n,i}^{(t,\theta_t)}\mathbf{Q}_{n,i}^{(t,\theta_t)}, n\in\overline{\mathcal{S}}_w^{(t)}.\label{eq:update}
    \end{align}
\end{definition}

Now let us prove the correctness, privacy and security of the ACSA-RW scheme. To proceed, we need the following lemmas.
\begin{lemma}\label{lemma:read}
    At any time $t,t\in\mathbb{N}$, User $t$ retrieves the desired submodel $\mathbf{W}_{\theta_t}^{(t-1)}$ from the answers returned by the servers $n,n\in\overline{\mathcal{S}}_r^{(t)}$ according to the ACSA-RW scheme while guaranteeing $T$-privacy.
\end{lemma}

\begin{proof}
    Let us consider the answers returned by the servers $n,n\in\overline{\mathcal{S}}_r^{(t)}$ in \eqref{eq:answer}. Note that for all $\ell\in[\#_r^{(t)}]$, $i\in[K_c]$, we have
    \begin{align}
         & \left(\mathbf{S}_n^{(t-1)}\right)^{\mathsf{T}}\mathbf{\Xi}_{n,\ell,i}^{(t)}\mathbf{Q}_{n,i}^{(t,\theta_t)}\notag                                                                                                                                                                                                                         \\
         & =\sum_{j\in[(\ell-1)R_r^{(t)}+1:\ell R_r^{(t)}]}\left(\sum_{i'\in[K_c]}\frac{1}{\alpha_n-f_{j,i'}}\dot{\mathbf{W}}^{(t-1)}_{j,i'}+\sum_{x\in[X]}\alpha_n^{x-1}\dot{\mathbf{Z}}_{j,x}^{(t-1)}\right)^{\mathsf{T}}\notag  \\
         & \hspace{4cm}\left(\frac{\prod_{i'\in[K_c]\setminus\{i\}}(\alpha_n-f_{j,i'})}{\prod_{i'\in[K_c]\setminus\{i\}}(f_{j,i}-f_{j,i'})}\right)\left(\mathbf{e}_K(\theta_t)+(\alpha_n-f_{j,i})\sum_{s\in[T]}\alpha_n^{s-1}\ddot{\mathbf{Z}}_{j,i,s}^{(t)}\right)                                                                                                                                                                                                         \\
         & =\sum_{j\in[(\ell-1)R_r^{(t)}+1:\ell R_r^{(t)}]}\left(\sum_{i'\in[K_c]}\frac{1}{\alpha_n-f_{j,i'}}\dot{\mathbf{W}}^{(t-1)}_{j,i'}+\sum_{x\in[X]}\alpha_n^{x-1}\dot{\mathbf{Z}}_{j,x}^{(t-1)}\right)^{\mathsf{T}}                                                                                                                                   \\
         & \hspace{3cm}\left(\left(\frac{\prod_{i'\in[K_c]\setminus\{i\}}(\alpha_n-f_{j,i'})}{\prod_{i'\in[K_c]\setminus\{i\}}(f_{j,i}-f_{j,i'})}\right)\mathbf{e}_K(\theta_t)+\frac{\prod_{i'\in[K_c]}(\alpha_n-f_{j,i})}{\prod_{i'\in[K_c]\setminus\{i\}}(f_{j,i}-f_{j,i'})}\sum_{s\in[T]}\alpha_n^{s-1}\ddot{\mathbf{Z}}_{j,i,s}^{(t)}\right) \\
         & =\sum_{j\in[(\ell-1)R_r^{(t)}+1:\ell R_r^{(t)}]}\left(\frac{1}{\alpha_n-f_{j,i}}\dot{\mathbf{W}}^{(t-1)}_{j,i}\mathbf{e}_K(\theta_t)+\sum_{m\in[X+T+K_c-1]}\alpha_n^{m-1}\dot{I}_{j,i,m}^{(t)}\right)                                                                                                                           \\
         & =\sum_{j\in[(\ell-1)R_r^{(t)}+1:\ell R_r^{(t)}]}\frac{1}{\alpha_n-f_{j,i}}W^{(t-1)}_{\theta_t}(j,i)+\sum_{m\in[X+T+K_c-1]}\alpha_n^{m-1}\ddot{I}_{r,i,m}^{(t)}.\label{eq:read1}
    \end{align}
    For all $j\in[J], m\in[X+T+K_c-1]$, $\dot{I}_{j,i,m}^{(t)}$ are various linear combinations of inner products of $\left(\dot{\mathbf{W}}^{(t-1)}_{j,i}\right)_{i\in[K_c]}, \mathbf{e}_K(\theta_t), \left(\dot{\mathbf{Z}}_{j,x}^{(t-1)}\right)_{x\in[X]}$ and $\left(\ddot{\mathbf{Z}}_{j,i,s}^{(t)}\right)_{i\in[K_c],s\in[T]}$, whose exact forms are irrelevant. Besides, $\ddot{I}_{\ell,i,m}^{(t)}=\sum_{j\in[(\ell-1)R_r^{(t)}+1:\ell R_r^{(t)}]}\dot{I}_{j,i,m}^{(t)}, \forall \ell\in[\#_r^{(t)}], i\in[K_c]$. Note that for all $j\in[J]$, $i\in[K_c]$, $\prod_{i'\in[K_c]\setminus\{i\}}(f_{j,i}-f_{j,i'})$ is the remainder of the polynomial division (with respect to $\alpha_n$)\\ $\left(\prod_{i'\in[K_c]\setminus\{i\}}(\alpha_n-f_{j,i'})\right)/\left(\alpha_n-f_{j,i}\right)$, which is for normalization. The existence of multiplicative inverse of $\prod_{i'\in[K_c]\setminus\{i\}}(f_{j,i}-f_{j,i'})$ is guaranteed by Proposition \ref{prop:distink}, i.e., $\prod_{i'\in[K_c]\setminus\{i\}}(f_{j,i}-f_{j,i'})\neq 0$. Therefore, due to the fact that the desired symbols are carried by the Cauchy terms (i.e., the first term in \eqref{eq:read1}) and the interference symbols (i.e., the undesired symbols, second term in \eqref{eq:read1}) are aligned along the Vandermonde terms, the desired symbols $\left(W_{\theta_t}^{(t-1)}(j,i)\right)_{j\in[(\ell-1)R_r^{(t)}+1:\ell R_r^{(t)}]}$ of the submodel $\mathbf{W}_{\theta_t}^{(t-1)}$ are resolvable by inverting the following Cauchy-Vandermonde matrix.
    \begin{align}
        \mathbf{C}=\begin{bmatrix}
            \frac{1}{f_{(\ell-1)R_r^{(t)}+1,i}-\alpha_{\overline{\mathcal{S}}_r^{(t)}(1)}}                                & \cdots & \frac{1}{f_{\ell R_r^{(t)},i}-\alpha_{\overline{\mathcal{S}}_r^{(t)}(1)}}                                & 1      & \cdots & \alpha_{\overline{\mathcal{S}}_r^{(t)}(1)}^{X+T-1}                                \\
            \frac{1}{f_{(\ell-1)R_r^{(t)}+1,i}-\alpha_{\overline{\mathcal{S}}_r^{(t)}(2)}}                                & \cdots & \frac{1}{f_{\ell R_r^{(t)},i}-\alpha_{\overline{\mathcal{S}}_r^{(t)}(2)}}                                & 1      & \cdots & \alpha_{\overline{\mathcal{S}}_r^{(t)}(2)}^{X+T-1}                                \\
            \vdots                                                                                                  & \vdots & \vdots                                                                                            & \vdots & \vdots & \vdots                                                                            \\
            \frac{1}{f_{(\ell-1)R_r^{(t)}+1,i}-\alpha_{\overline{\mathcal{S}}_r^{(t)}(|\overline{\mathcal{S}}_r^{(t)}|)}} & \cdots & \frac{1}{f_{\ell R_r^{(t)},i}-\alpha_{\overline{\mathcal{S}}_r^{(t)}(|\overline{\mathcal{S}}_r^{(t)}|)}} & 1      & \cdots & \alpha_{\overline{\mathcal{S}}_r^{(t)}(|\overline{\mathcal{S}}_r^{(t)}|)}^{X+T-1} \\
        \end{bmatrix}.\label{eq:cvmat}
    \end{align}
    It is remarkable that the non-singularity of the matrix $\mathbf{C}$ follows from the determinant of Cauchy-Vandermonde matrix (see, e.g., \cite{Gasca_Martinez_Muhlbach}) and the fact that according to Proposition \ref{prop:distinj}, the constants \\$\left((f_{j,i})_{j\in[(\ell-1)R_r^{(t)}+1:\ell R_r^{(t)}]}, (\alpha_n)_{n\in\overline{\mathcal{S}}_r^{(t)}}\right)$ are distinct for all $\ell\in[\#_r^{(t)}]$, regardless of the realizations of $\overline{\mathcal{S}}_r^{(t)}$ and $\overline{\mathcal{S}}_w^{(t)}$. Recall that $\overline{\mathcal{S}}_r^{(t)}(1)$, $\overline{\mathcal{S}}_r^{(t)}(2)$, etc., refer to distinct elements of $\overline{\mathcal{S}}_r^{(t)}$ (arranged in ascending order).  Therefore, the user is able to reconstruct the desired submodel due to the fact that $\mathbf{W}_{\theta_t}^{(t-1)}=\left(\left(W_{\theta_t}^{(t-1)}(j,i)\right)_{j\in[(\ell-1)R_r^{(t)}+1:\ell R_r^{(t)}]}\right)_{\ell\in[\#_r^{(t)}],i\in[K_c]}$. To see why the $T$-privacy holds, we note that by the construction of the query $\left(Q_{n}^{(t,\theta_t)}\right)_{n\in[N]}$, the vector $\mathbf{e}_K(\theta_t)$, which carries the information of desired index $\theta_t$, is protected by the MDS$(N,T)$ coded uniform i.i.d. random noise vectors. Thus the queries for the $N$ servers form a secret sharing of threshold $T$, and are independent of the queries and the increments of all prior users $\tau, \tau\in[t-1]$. Thus $T$-privacy is guaranteed. To calculate the download cost, we note that a total of $L=J K_c$ symbols of the desired submodel are retrieved out of the $K_c\#_r^{(t)}|\overline{\mathcal{S}}_r^{(t)}|$ downloaded symbols. Therefore we have $D_t=\left(K_c\#_r^{(t)}|\overline{\mathcal{S}}_r^{(t)}|\right)/(JK_c)=\left(\#_r^{(t)}\left(N-|\mathcal{S}_r^{(t)}|\right)\right)/J=\left(N-|\mathcal{S}_r^{(t)}|\right)/R_r^{(t)}=\left(N-|\mathcal{S}_r^{(t)}|\right)/\left(S_r^{\text{\normalfont thresh}}-|\mathcal{S}_r^{(t)}|\right)$. This completes the proof of Lemma \ref{lemma:read}.
\end{proof}
\begin{lemma}\label{lemma:write}
    At any time $t,t\in\mathbb{N}$, User $t$ correctly updates the desired submodel $\mathbf{W}_{\theta_t}^{(t-1)}$ and achieves  ACSA storage by uploading the ACSA increments to the servers $n,n\in\overline{\mathcal{S}}_w^{(t)}$ and exploiting the update equation \eqref{eq:update} according to the ACSA-RW scheme while guaranteeing $T$-privacy and $X_{\Delta}$-security.
\end{lemma}
\begin{proof}
    Let us first inspect the second term on the RHS of \eqref{eq:update}. Note that for any $t\in\mathbb{N}$, for all $n\in[N],i\in[K_c]$, we can write
    \begin{align}
         & \mathbf{\Omega}_{n,i}^{(t)}\mathbf{\Upsilon}_{n,i}^{(t)}\mathbf{P}_{n,i}^{(t,\theta_t)}\mathbf{Q}_{n,i}^{(t,\theta_t)}\notag                                                                                       \\
         & =\left[\left(\mathbf{\Gamma}_{n,1,i}^{(t)}\right)^{\mathsf{T}},\left(\mathbf{\Gamma}_{n,2,i}^{(t)}\right)^{\mathsf{T}},\cdots,\left(\mathbf{\Gamma}_{n,\#_w^{(t)},i}^{(t)}\right)^{\mathsf{T}}\right]^{\mathsf{T}},
    \end{align}
    where for all $\ell\in[\#_w^{(t)}]$, we define $\phi_{\ell}=(\ell-1)R_w^{(t)}$, and
    \begin{align}
         & \mathbf{\Gamma}_{n,\ell,i}^{(t)}\notag                         \\
         & =\left[\begin{array}{l}
                \left(\frac{\prod_{j\in\mathcal{F}_{\phi_{\ell}+1}^{(t)}}(\alpha_n-f_{j,i})}{\prod_{j\in\mathcal{F}_{\phi_{\ell}+1}^{(t)}}(f_{\phi_{\ell}+1,i}-f_{j,i})}\right)\left(\sum_{j\in[\phi_{\ell}+1:\phi_{\ell}+R_w^{(t)}]}\frac{1}{\alpha_n-f_{j,i}}\Delta_{j,i}^{(t)}+\sum_{x\in[X_{\Delta}]}\alpha_n^{x-1}\dddot{Z}_{\ell,i,x}^{(t)}\right)             \\
                \hspace{1cm}\left(\frac{\prod_{m\in\mathcal{S}_w^{(t)}}(\alpha_n-\alpha_m)}{\prod_{m\in\mathcal{S}_w^{(t)}}(f_{\phi_{\ell}+1,i}-\alpha_m)}\right)\left(\mathbf{e}_K(\theta_t)+(\alpha_n-f_{\phi_{\ell}+1,i})\sum_{s\in[T]}\alpha_n^{s-1}\ddot{\mathbf{Z}}_{\phi_{\ell}+1,i,s}^{(t)}\right)                                              \\
                \hdashline[2pt/2pt]
                \left(\frac{\prod_{j\in\mathcal{F}_{\phi_{\ell}+2}^{(t)}}(\alpha_n-f_{j,i})}{\prod_{j\in\mathcal{F}_{\phi_{\ell}+2}^{(t)}}(f_{\phi_{\ell}+2,i}-f_{j,i})}\right)\left(\sum_{j\in[\phi_{\ell}+1:\phi_{\ell}+R_w^{(t)}]}\frac{1}{\alpha_n-f_{j,i}}\Delta_{j,i}^{(t)}+\sum_{x\in[X_{\Delta}]}\alpha_n^{x-1}\dddot{Z}_{\ell,i,x}^{(t)}\right)             \\
                \hspace{1cm}\left(\frac{\prod_{m\in\mathcal{S}_w^{(t)}}(\alpha_n-\alpha_m)}{\prod_{m\in\mathcal{S}_w^{(t)}}(f_{\phi_{\ell}+2,i}-\alpha_m)}\right)\left(\mathbf{e}_K(\theta_t)+(\alpha_n-f_{\phi_{\ell}+2,i})\sum_{s\in[T]}\alpha_n^{s-1}\ddot{\mathbf{Z}}_{\phi_{\ell}+2,i,s}^{(t)}\right)                                              \\
                \hdashline[2pt/2pt]
                \multicolumn{1}{c}{\vdots}                                                                                                                                                                                                                                                                                               \\
                \hdashline[2pt/2pt]
                \left(\frac{\prod_{j\in\mathcal{F}_{\phi_{\ell}+R_w^{(t)}}^{(t)}}(\alpha_n-f_{j,i})}{\prod_{j\in\mathcal{F}_{\phi_{\ell}+R_w^{(t)}}^{(t)}}(f_{\phi_{\ell}+R_w^{(t)},i}-f_{j,i})}\right)\left(\sum_{j\in[\phi_{\ell}+1:\phi_{\ell}+R_w^{(t)}]}\frac{1}{\alpha_n-f_{j,i}}\Delta_{j,i}^{(t)}+\sum_{x\in[X_{\Delta}]}\alpha_n^{x-1}\dddot{Z}_{\ell,i,x}^{(t)}\right) \\
                \hspace{1cm}\left(\frac{\prod_{m\in\mathcal{S}_w^{(t)}}(\alpha_n-\alpha_m)}{\prod_{m\in\mathcal{S}_w^{(t)}}(f_{\phi_{\ell}+R_w^{(t)},i}-\alpha_m)}\right)\left(\mathbf{e}_K(\theta_t)+(\alpha_n-f_{\phi_{\ell}+R_w^{(t)},i})\sum_{s\in[T]}\alpha_n^{s-1}\ddot{\mathbf{Z}}_{\phi_{\ell}+R_w^{(t)},i,s}^{(t)}\right)
            \end{array}\right]\label{eq:update0}  \\
         & =\left[\begin{array}{l}
                \left(\frac{\prod_{m\in\mathcal{S}_w^{(t)}}(\alpha_n-\alpha_m)}{\prod_{m\in\mathcal{S}_w^{(t)}}(f_{\phi_{\ell}+1,i}-\alpha_m)}\right)\left(\frac{1}{\alpha_n-f_{\phi_{\ell}+1,i}}\Delta_{\phi_{\ell}+1,i}^{(t)}+\sum_{m\in[X_{\Delta}+R_w^{(t)}-1]}\alpha_n^{m-1}\dddot{I}_{\phi_{\ell}+1,i,m}^{(t)}\right)                 \\
                \hspace{1cm}\left(\mathbf{e}_K(\theta_t)+(\alpha_n-f_{\phi_{\ell}+1,i})\sum_{s\in[T]}\alpha_n^{s-1}\ddot{\mathbf{Z}}_{\phi_{\ell}+1,i,s}^{(t)}\right)                                                                                                                                                         \\
                \hdashline[2pt/2pt]
                \left(\frac{\prod_{m\in\mathcal{S}_w^{(t)}}(\alpha_n-\alpha_m)}{\prod_{m\in\mathcal{S}_w^{(t)}}(f_{\phi_{\ell}+2,i}-\alpha_m)}\right)\left(\frac{1}{\alpha_n-f_{\phi_{\ell}+2,i}}\Delta_{\phi_{\ell}+2,i}^{(t)}+\sum_{m\in[X_{\Delta}+R_w^{(t)}-1]}\alpha_n^{m-1}\dddot{I}_{\phi_{\ell}+2,i,m}^{(t)}\right)                 \\
                \hspace{1cm}\left(\mathbf{e}_K(\theta_t)+(\alpha_n-f_{\phi_{\ell}+2,i})\sum_{s\in[T]}\alpha_n^{s-1}\ddot{\mathbf{Z}}_{\phi_{\ell}+2,i,s}^{(t)}\right)                                                                                                                                                         \\
                \hdashline[2pt/2pt]
                \multicolumn{1}{c}{\vdots}                                                                                                                                                                                                                                                                          \\
                \hdashline[2pt/2pt]
                \left(\frac{\prod_{m\in\mathcal{S}_w^{(t)}}(\alpha_n-\alpha_m)}{\prod_{m\in\mathcal{S}_w^{(t)}}(f_{\phi_{\ell}+R_w^{(t)},i}-\alpha_m)}\right)\left(\frac{1}{\alpha_n-f_{\phi_{\ell}+R_w^{(t)},i}}\Delta_{\phi_{\ell}+R_w^{(t)},i}^{(t)}+\sum_{m\in[X_{\Delta}+R_w^{(t)}-1]}\alpha_n^{m-1}\dddot{I}_{\phi_{\ell}+R_w^{(t)},i,m}^{(t)}\right) \\
                \hspace{1cm}\left(\mathbf{e}_K(\theta_t)+(\alpha_n-f_{\phi_{\ell}+R_w^{(t)},m})\sum_{s\in[T]}\alpha_n^{s-1}\ddot{\mathbf{Z}}_{\phi_{\ell}+R_w^{(t)},i,s}^{(t)}\right)
            \end{array}\right]\label{eq:update1}  \\
         & =\left[\begin{array}{l}
                \left(\frac{1}{\alpha_n-f_{\phi_{\ell}+1,i}}\Delta_{\phi_{\ell}+1,i}^{(t)}+\sum_{m\in[X_{\Delta}+R_w^{(t)}+|\mathcal{S}_w^{(t)}|-1]}\alpha_n^{m-1}\ddddot{I}_{\phi_{\ell}+1,i,m}^{(t)}\right)             \\
                \hspace{1cm}\left(\mathbf{e}_K(\theta_t)+(\alpha_n-f_{\phi_{\ell}+1,i})\sum_{s\in[T]}\alpha_n^{s-1}\ddot{\mathbf{Z}}_{\phi_{\ell}+1,i,s}^{(t)}\right)                                \\
                \hdashline[2pt/2pt]
                \left(\frac{1}{\alpha_n-f_{\phi_{\ell}+2,i}}\Delta_{\phi_{\ell}+2,i}^{(t)}+\sum_{m\in[X_{\Delta}+R_w^{(t)}+|\mathcal{S}_w^{(t)}|-1]}\alpha_n^{m-1}\ddddot{I}_{\phi_{\ell}+2,i,m}^{(t)}\right)             \\
                \hspace{1cm}\left(\mathbf{e}_K(\theta_t)+(\alpha_n-f_{\phi_{\ell}+2,i})\sum_{s\in[T]}\alpha_n^{s-1}\ddot{\mathbf{Z}}_{\phi_{\ell}+2,i,s}^{(t)}\right)                                \\
                \hdashline[2pt/2pt]
                \multicolumn{1}{c}{\vdots}                                                                                                                                                 \\
                \hdashline[2pt/2pt]
                \left(\frac{1}{\alpha_n-f_{\phi_{\ell}+R_w^{(t)},i}}\Delta_{\phi_{\ell}+R_w^{(t)},i}^{(t)}+\sum_{m\in[X_{\Delta}+R_w^{(t)}+|\mathcal{S}_w^{(t)}|-1]}\alpha_n^{m-1}\ddddot{I}_{\phi_{\ell}+R_w^{(t)},i,m}^{(t)}\right) \\
                \hspace{1cm}\left(\mathbf{e}_K(\theta_t)+(\alpha_n-f_{\phi_{\ell}+R_w^{(t)},i})\sum_{s\in[T]}\alpha_n^{s-1}\ddot{\mathbf{Z}}_{\phi_{\ell}+R_w^{(t)},i,s}^{(t)}\right)
            \end{array}\right]\label{eq:update2}  \\
         & =\left[\begin{array}{l}
                \left(\frac{1}{\alpha_n-f_{\phi_{\ell}+1,i}}\Delta_{\phi_{\ell}+1,i}^{(t)}\mathbf{e}_K(\theta_t)+\sum_{m\in[X_{\Delta}+R_w^{(t)}+|\mathcal{S}_w^{(t)}|+T-1]}\alpha_n^{m-1}\dot{\mathbf{I}}_{\phi_{\ell}+1,i,m}^{(t)}\right) \\
                \hdashline[2pt/2pt]
                \left(\frac{1}{\alpha_n-f_{\phi_{\ell}+2,i}}\Delta_{\phi_{\ell}+2,i}^{(t)}\mathbf{e}_K(\theta_t)+\sum_{m\in[X_{\Delta}+R_w^{(t)}+|\mathcal{S}_w^{(t)}|+T-1]}\alpha_n^{m-1}\dot{\mathbf{I}}_{\phi_{\ell}+2,i,m}^{(t)}\right) \\
                \hdashline[2pt/2pt]
                \multicolumn{1}{c}{\vdots}                                                                                                                                                                   \\
                \hdashline[2pt/2pt]
                \left(\frac{1}{\alpha_n-f_{\phi_{\ell}+R_w^{(t)},i}}\Delta_{\phi_{\ell}+R_w^{(t)},i}^{(t)}\mathbf{e}_K(\theta_t)+\sum_{m\in[X_{\Delta}+R_w^{(t)}+|\mathcal{S}_w^{(t)}|+T-1]}\alpha_n^{m-1}\dot{\mathbf{I}}_{\phi_{\ell}+R_w^{(t)},i,m}^{(t)}\right)
            \end{array}\right],\label{eq:update3}
    \end{align}
    where for all $v\in[R_w^{(t)}]$, $\left(\dddot{I}_{\phi_{\ell}+v,m}\right)_{m\in[X_{\Delta}+R_w^{(t)}-1]}$, $\left(\ddddot{I}_{\phi_{\ell}+v,m}\right)_{m\in[X_{\Delta}+R_w^{(t)}+|\mathcal{S}_w^{(t)}|-1]}$ and\\ $\left(\dot{\mathbf{I}}_{\phi_{\ell}+v,m}\right)_{m\in[X_{\Delta}+R_w^{(t)}+|\mathcal{S}_w^{(t)}|+T-1]}$ are various interference symbols, whose exact forms are not important. Note that in \eqref{eq:update1}, we multiply the first two terms in each row of \eqref{eq:update0}. It can be justified from the fact that the constants $\left(f_{j,i}\right)_{j\in[\phi_{\ell}+1:\phi_{\ell}+R_w^{(t)}]}$ are distinct according to Proposition \ref{prop:distinj}. Besides, for all $v\in[R_w^{(t)}]$, according to the definition of $\mathcal{F}_{\phi_{\ell}+v}^{(t)}$, we can equivalently write $\mathcal{F}_{\phi_{\ell}+v}^{(t)}=[\phi_{\ell}+1:\phi_{\ell}+R_w^{(t)}]\setminus\{\phi_{\ell}+v\}$, thus $\left|\mathcal{F}_{\phi_{\ell}+v}^{(t)}\right|=R_w^{(t)}-1$ and  $(\phi_{\ell}+v)\not\in \mathcal{F}_{\phi_{\ell}+v}$. Note that the denominator in the first term in each row of \eqref{eq:update0} is for normalization, which is the remainder of the polynomial division (with respect to $\alpha_n$) $\left(\prod_{j\in\mathcal{F}_{\phi_{\ell}+v}^{(t)}}(\alpha_n-f_{j,i})\right)/(\alpha_n-f_{\phi_{\ell}+v,i})$. In \eqref{eq:update2}, we multiply the first two terms in each row of \eqref{eq:update1}, and it can be justified by noting that the denominator in the first term in each row of \eqref{eq:update1} is the remainder of the polynomial division (with respect to $\alpha_n$) $\left(\prod_{m\in\mathcal{S}_w^{(t)}}(\alpha_n-\alpha_m)\right)/(\alpha_n-f_{\phi_{\ell}+v,i})$. And finally in \eqref{eq:update3}, we multiply the two terms in each row of \eqref{eq:update2}. Therefore, for all $n\in[N],i\in[K_c]$, the term $\mathbf{\Omega}_{n,i}^{(t)}\mathbf{\Upsilon}_{n,i}^{(t)}\mathbf{P}_{n,i}^{(t,\theta_t)}\mathbf{Q}_{n,i}^{(t,\theta_t)}$ can be written as follows.
    \begin{align}
         & \mathbf{\Omega}_{n,i}^{(t)}\mathbf{\Upsilon}_{n,i}^{(t)}\mathbf{P}_{n,i}^{(t,\theta_t)}\mathbf{Q}_{n,i}^{(t,\theta_t)}\notag \\
         & =\left[\begin{array}{l}
                \frac{1}{\alpha_n-f_{1,i}}\Delta_{1,i}^{(t)}\mathbf{e}_K(\theta_t)+\sum_{m\in[X_{\Delta}+R_w^{(t)}+|\mathcal{S}_w^{(t)}|+T-1]}\alpha_n^{m-1}\dot{\mathbf{I}}_{1,i,m}^{(t)} \\
                \hdashline[2pt/2pt]
                \frac{1}{\alpha_n-f_{2,i}}\Delta_{2,i}^{(t)}\mathbf{e}_K(\theta_t)+\sum_{m\in[X_{\Delta}+R_w^{(t)}+|\mathcal{S}_w^{(t)}|+T-1]}\alpha_n^{m-1}\dot{\mathbf{I}}_{2,i,m}^{(t)} \\
                \hdashline[2pt/2pt]
                \multicolumn{1}{c}{\vdots}                                                                                                                                 \\
                \hdashline[2pt/2pt]
                \frac{1}{\alpha_n-f_{J,i}}\Delta_{J,i}^{(t)}\mathbf{e}_K(\theta_t)+\sum_{m\in[X_{\Delta}+R_w^{(t)}+|\mathcal{S}_w^{(t)}|+T-1]}\alpha_n^{m-1}\dot{\mathbf{I}}_{J,i,m}^{(t)}
            \end{array}\right].
    \end{align}
    Note that $X=X_{\Delta}+R_w^{(t)}+|\mathcal{S}_w^{(t)}|+T-1$, the update equation \eqref{eq:update} is thus correct because
    \begin{align}
         & \mathbf{S}_n^{(t-1)}+\sum_{i\in[K_c]}\mathbf{\Omega}_{n,i}^{(t)}\mathbf{\Upsilon}_{n,i}^{(t)}\mathbf{P}_{n,i}^{(t,\theta_t)}\mathbf{Q}_{n,i}^{(t,\theta_t)}\notag \\
         & =\left[\begin{array}{c}
                \sum_{i\in[K_c]}\frac{1}{\alpha_n-f_{1,i}}\dot{\mathbf{W}}^{(t-1)}_{1,i}+\sum_{x\in[X]}\alpha_n^{x-1}\dot{\mathbf{Z}}_{1,x}^{(t-1)} \\
                \hdashline[2pt/2pt]
                \sum_{i\in[K_c]}\frac{1}{\alpha_n-f_{2,i}}\dot{\mathbf{W}}^{(t-1)}_{2,i}+\sum_{x\in[X]}\alpha_n^{x-1}\dot{\mathbf{Z}}_{2,x}^{(t-1)} \\
                \hdashline[2pt/2pt]
                \vdots                                                                                                                              \\
                \hdashline[2pt/2pt]
                \sum_{i\in[K_c]}\frac{1}{\alpha_n-f_{J,i}}\dot{\mathbf{W}}^{(t-1)}_{J,i}+\sum_{x\in[X]}\alpha_n^{x-1}\dot{\mathbf{Z}}_{J,x}^{(t-1)}
            \end{array}\right]\notag                                                                                                                    \\
         & \quad+\left[\begin{array}{l}
                \sum_{i\in[K_c]}\frac{1}{\alpha_n-f_{1,i}}\Delta_{1,i}^{(t)}\mathbf{e}_K(\theta_t)+\sum_{x\in[X]}\alpha_n^{x-1}\ddot{\mathbf{I}}_{1,x}^{(t)} \\
                \hdashline[2pt/2pt]
                \sum_{i\in[K_c]}\frac{1}{\alpha_n-f_{2,i}}\Delta_{2,i}^{(t)}\mathbf{e}_K(\theta_t)+\sum_{x\in[X]}\alpha_n^{x-1}\ddot{\mathbf{I}}_{2,x}^{(t)} \\
                \hdashline[2pt/2pt]
                \multicolumn{1}{c}{\vdots}                                                                                                                   \\
                \hdashline[2pt/2pt]
                \sum_{i\in[K_c]}\frac{1}{\alpha_n-f_{J,i}}\Delta_{J,i}^{(t)}\mathbf{e}_K(\theta_t)+\sum_{x\in[X]}\alpha_n^{x-1}\ddot{\mathbf{I}}_{J,x}^{(t)}
            \end{array}\right]                                                                                                                     \\
         & =\mathbf{S}_n^{(t)},
    \end{align}
    where for all $t\in\mathbb{N}, j\in[J], x\in[X]$, $\dot{\mathbf{Z}}_{j,x}^{(t)}=\dot{\mathbf{Z}}_{j,x}^{(t-1)}+\ddot{\mathbf{I}}_{j,x}^{(t)}$, and $\ddot{\mathbf{I}}_{j,x}^{(t)}=\sum_{i\in[K_c]}\dot{\mathbf{I}}_{j,i,x}^{(t)}$. Note that for all $n\in\mathcal{S}_w^{(t)}$, $i\in[K_c]$, we have $\mathbf{\Omega}_{n,i}^{(t)}=\mathbf{0}$. Therefore, for all $n\in\mathcal{S}_w^{(t)}$, it holds that $\mathbf{S}_n^{(t)}=\mathbf{S}_n^{(t-1)}$. In other words, the update equation \eqref{eq:update} correctly updates the desired submodel and achieves the ACSA storage at time $t+1$ by updating the storage of server $n, n\in\overline{\mathcal{S}}_w^{(t)}$. The proof of $T$-privacy follows from that in Lemma \ref{lemma:read}, so we do not repeat it here. The proof of $X_{\Delta}$-security follows from the fact that by the definition of ACSA increment $\left(P_n^{(t,\theta_t)}\right)_{n\in[N]}$, the symbols of the increment $\mathbf{\Delta}_t$ are protected by the MDS$(N,X_{\Delta})$ coded uniform i.i.d. random noise symbols. Thus the ACSA increment for the $N$ servers form a secret sharing of threshold $X_{\Delta}$, and it is independent of the write-queries by the users $\tau, \tau\in[t-1]$ and the read-queries by the users $\tau, \tau\in[t]$. Finally let us calculate the upload cost. The upload cost consists of two parts, i.e., the upload cost of the ACSA query and the upload cost of the ACSA increment. Note that to upload the ACSA query, a total of $\left|[N]\setminus\mathcal{S}_r^{(t)}\setminus\mathcal{S}_w^{(t)}\right|\mu K K_c$ $q$-ary symbols must be uploaded. On the other hand, to upload the ACSA increment, we need to upload a total of $\left(N-|\mathcal{S}_w^{(t)}|\right)\#_w^{(t)} K_c$ $q$-ary symbols. Therefore, in the limit as $L/K\rightarrow\infty$, the normalized upload cost is
    \begin{align}
        U_t & =\frac{K_c\left(N-|\mathcal{S}_w^{(t)}|\right)\#_w^{(t)}+\left|[N]\setminus\mathcal{S}_r^{(t)}\setminus\mathcal{S}_w^{(t)}\right|\mu KK_c}{L}                    \\
            & \stackrel{L/K\rightarrow\infty}{=}\frac{N-|\mathcal{S}_w^{(t)}|}{R_w^{(t)}}\\
            &=\frac{N-|\mathcal{S}_w^{(t)}|}{S_w^{\text{\normalfont thresh}}-|\mathcal{S}_w^{(t)}|}.
    \end{align}
    This completes the proof of Lemma \ref{lemma:write}.
\end{proof}

The ACSA-RW scheme satisfies the correctness, $T$-privacy, and $X_{\Delta}$-security constraints for each update $t, t\in\mathbb{N}$ because of Lemma \ref{lemma:read} and Lemma \ref{lemma:write}, which hold for all $t, t\in\mathbb{N}$. Now let us see why the $X$-security constraint is satisfied. By the definition of the ACSA storage (i.e., Definition \ref{def:ACSAstor}) at any time $t,t\in\mathbb{N}$, the symbols of the $K$ submodels are protected by the MDS$(N,X)$ coded i.i.d. uniform random noise symbols. In other words, it forms a secret sharing of threshold $X$, thus $X$-security is guaranteed.
\begin{remark}{\bf(Byzantine Tolerance) }It is remarkable that the answers returned by the servers in the private read phase can be viewed as codewords of an MDS code (generated by the Cauchy-Vandermonde matrix, see, e.g., \cite{Jia_Jafar_MDSXSTPIR}). Therefore, with additional $2B$ redundant answers from the servers, we can correct up to $B$ erroneous answers.
\end{remark}
\begin{remark}{\bf(Symmetric Security) }If common randomness is allowed among the servers, so-called {\bf symmetric security} can be achieved\cite{Chen_Jia_Wang_Jafar}, i.e., the user will learn nothing about the global model beyond the desired submodel. Note that this does not affect the communication cost of the ACSA-RW scheme.
\end{remark}
\begin{remark}{\bf(External Adversaries versus Internal Adversaries) }\label{remark:extintadv}Recall that the $X$-security constraint only requires protection against  an external adversary who can  access the current storage  but not the past history at any $X$-servers. However, a closer look at the ACSA-RW scheme reveals that if the number of compromised servers is no more than $\min(X_\Delta, T)$, then even an \emph{internal} adversary, i.e., an adversary who has access to the entire history of all previous stored values and queries seen by the compromised servers, can still learn nothing about the stored submodels.
\end{remark}
\begin{remark}\label{remark:ac}{\bf(Access Complexity) }We define the access complexity as the number of elements over the finite field $\mathbb{F}_q$ that must be accessed/updated during the private read and private write phases. We note that at any time $t,t\in\mathbb{N}$, the access complexity of each of the responsive servers in the private read and write phases is at most $KL/K_c$. Hence with greater $K_c$, it is possible to reduce the access complexity.
\end{remark}
\begin{remark}{\bf(Encoding and Decoding Complexity)} Let us consider the complexity of the encoding and decoding algorithms of our construction. It is worth noting that the computations for producing the ACSA storage, ACSA query and ACSA increment can be regarded as multiplications of (scaled) Cauchy-Vandermonde matrices with various vectors. The computation for recovering the desired submodel by the user from the answers of the $N$ servers can be viewed as solving linear systems defined by Cauchy-Vandermonde matrices. Cauchy-Vandermonde matrices are an important class of structured matrices, for which ``\emph{superfast}'' algorithms have been studied extensively \cite{Pan_Structured,Finck_FastCV}. Therefore, by these superfast algorithms, the complexity of producing the ACSA storage, ACSA query and ACSA increment is at most $\widetilde{\mathcal{O}}((LKN\log^2N)/K_c), \widetilde{\mathcal{O}}(\mu KNK_c\log^2N)$ and $\widetilde{\mathcal{O}}(U_tL\log^2N)$, respectively. On the other hand, the complexity of decoding the desired submodel from the answers of the servers is at most $\widetilde{\mathcal{O}}(D_tL\log^2N)$. It is obvious that the encoding/decoding algorithms have a complexity that is almost linear in their output/input sizes.
\end{remark}

\subsection{Example}\label{sec:example}
Let us consider an illustrative example to make the construction of the ACSA-RW scheme and the proof of Theorem \ref{thm:acsarw} more accessible. In particular, for this example, let us set $N=8, X=4, T=1, X_{\Delta}=1$ and $K_c=1$ (note that $L=J$ in this case). We follow the notations used in the proof of Theorem \ref{thm:acsarw}, and since $K_c=1$, we are able to omit an index in some of the subscripts, which corresponds to $i\in[K_c]$. For example, for all $l\in[L]$, $f_{l,1}$ is abbreviated as $f_{l}$, etc. We have $S_r^{\text{\normalfont thresh}}=3, S_w^{\text{\normalfont thresh}}=3$. Let $\xi$ be a positive integer, and we set $L=\xi\cdot\lcm\left([S_r^{\text{\normalfont thresh}}]\cup[S_w^{\text{\normalfont thresh}}]\right)=6\xi$, i.e., each of the $K$ submodels consists of $L=6\xi$ symbols from a finite field $\mathbb{F}_q, q\geq N+\mu=11$. Let $\alpha_1,\alpha_2,\cdots,\alpha_8, \widetilde{f}_1,\widetilde{f}_2,\widetilde{f}_3$ be a total of $11$ distinct elements from the finite field $\mathbb{F}_q$. For $\left(f_{l}\right)_{l\in[L]}$, let us define
\begin{align}
    (f_1,f_2,\cdots,f_{L})=(\widetilde{f}_1,\widetilde{f}_2,\widetilde{f}_3,\widetilde{f}_1,\widetilde{f}_2,\widetilde{f}_3,\cdots,\widetilde{f}_1,\widetilde{f}_2,\widetilde{f}_3).
\end{align}
Let $\left(\dot{\mathbf{Z}}_{l,x}^{(0)}\right)_{l\in[L],x\in[4]}$ be uniformly i.i.d. column vectors from $\mathbb{F}_q^K$. The initial ACSA storage at the $N=8$ servers is defined as follows.
\begin{align}
    \mathbf{S}_n^{(0)}=
    \begin{bmatrix}
        \frac{1}{\alpha_n-f_1}\dot{\mathbf{W}}^{(0)}_{1}+\sum_{x\in[4]}\alpha_n^{x-1}\dot{\mathbf{Z}}_{1,x}^{(0)} \\
        \frac{1}{\alpha_n-f_2}\dot{\mathbf{W}}^{(0)}_{2}+\sum_{x\in[4]}\alpha_n^{x-1}\dot{\mathbf{Z}}_{2,x}^{(0)} \\
        \vdots                                                                                                    \\
        \frac{1}{\alpha_n-f_{L}}\dot{\mathbf{W}}^{(0)}_{L}+\sum_{x\in[4]}\alpha_n^{x-1}\dot{\mathbf{Z}}_{L,x}^{(0)}
    \end{bmatrix},\forall n\in[8].
\end{align}
Now let us assume that User $1$ experiences $|\mathcal{S}_r^{(1)}|=1, |\mathcal{S}_w^{(1)}|=2$, i.e., there is $1$ dropout server in the read phase and $2$ dropout servers in the write phase. Note that the two sets  can be arbitrarily realized. We have $R_r^{(1)}=2, R_w^{(1)}=1, \#_r^{(1)}=3\xi, \#_w^{(1)}=6\xi$. The ACSA query sent by User $1$ to the $n^{th}$ server, $n\in[N]$ is defined as follows.
\begin{align}
    \mathbf{Q}^{(1,\theta_1)}_{n}=\begin{bmatrix}
        \mathbf{e}_K(\theta_1)+(\alpha_n-f_1)\ddot{\mathbf{Z}}_{1,1}^{(1)} \\
        \mathbf{e}_K(\theta_1)+(\alpha_n-f_2)\ddot{\mathbf{Z}}_{2,1}^{(1)} \\
        \vdots                                                             \\
        \mathbf{e}_K(\theta_1)+(\alpha_n-f_{L})\ddot{\mathbf{Z}}_{L,1}^{(1)}
    \end{bmatrix}.
\end{align}
where 
\begin{align}
     & \left(\ddot{\mathbf{Z}}_{1,1}^{(1)},\ddot{\mathbf{Z}}_{2,1}^{(1)},\cdots,\ddot{\mathbf{Z}}_{L,1}^{(1)}\right)=\left(\widetilde{\mathbf{Z}}_{1,1}^{(1)},\widetilde{\mathbf{Z}}_{2,1}^{(1)},\widetilde{\mathbf{Z}}_{3,1}^{(1)},\widetilde{\mathbf{Z}}_{1,1}^{(1)},\widetilde{\mathbf{Z}}_{2,1}^{(1)},\widetilde{\mathbf{Z}}_{3,1}^{(1)},\cdots,\widetilde{\mathbf{Z}}_{1,1}^{(1)},\widetilde{\mathbf{Z}}_{2,1}^{(1)},\widetilde{\mathbf{Z}}_{3,1}^{(1)}\right),
\end{align}
and $\left(\widetilde{\mathbf{Z}}_{l,1}^{(1)}\right)_{l\in[3]}$ are i.i.d. uniform column vectors from $\mathbb{F}_q^K$. Evidently, for all $n\in[8]$, $\mathbf{Q}^{(1,\theta_1)}_{n}$ is uniquely determined by its first $3K$ rows. Upon receiving the queries, servers $n,n\in\overline{\mathcal{S}}_r^{(1)}$ respond to User $1$ with answers $A_n^{(1,\theta_1)}$ constructed as follows.
\begin{align}
    A_n^{(1,\theta_1)}=\left(\mathbf{A}_{n,\ell}^{(1,\theta_1)}\right)_{\ell\in[3\xi]}, n\in\overline{\mathcal{S}}_r^{(1)},
\end{align}
where for all $n\in\overline{\mathcal{S}}_r^{(1)}, \ell\in[3\xi]$,
\begin{align}
    \mathbf{A}_{n,\ell}^{(1,\theta_1)} & =\left(\mathbf{S}_n^{(0)}\right)^{\mathsf{T}}\mathbf{\Xi}_{n,\ell}^{(1)}\mathbf{Q}_{n}^{(1,\theta_1)}                                                                                                                                      \\
                                    & =\sum_{l\in[2\ell-1:2\ell]}\left(\frac{1}{\alpha_n-f_l}\dot{\mathbf{W}}^{(0)}_{l}+\sum_{x\in[4]}\alpha_n^{x-1}\dot{\mathbf{Z}}_{l,x}^{(0)}\right)^{\mathsf{T}}\left(\mathbf{e}_K(\theta_1)+(\alpha_n-f_l)\ddot{\mathbf{Z}}_{l,1}^{(1)}\right) \\
                                    & =\frac{W_{\theta_1}^{(0)}(2\ell-1)}{\alpha_n-f_{2\ell-1}}+\frac{W_{\theta_1}^{(0)}(2\ell)}{\alpha_n-f_{2\ell}}+\sum_{i\in[5]}\alpha_n^{i-1}\ddot{I}_{\ell,i}^{(1)},
\end{align}
where $\left(\ddot{I}_{\ell,i}\right)_{\ell\in[3\xi],i\in[5]}$ are various interference symbols. Thus for all $\ell\in[3\xi]$, User $1$ recovers the desired symbols $W_{\theta_1}^{(0)}(2\ell-1)$ and $W_{\theta_1}^{(0)}(2\ell)$ by inverting the following matrix.
\begin{align}
    \mathbf{C}=\begin{bmatrix}
        \frac{1}{f_{2\ell-1}-\alpha_{\overline{\mathcal{S}}_r^{(1)}(1)}} & \frac{1}{f_{2\ell}-\alpha_{\overline{\mathcal{S}}_r^{(1)}(1)}} & 1      & \cdots & \alpha_{\overline{\mathcal{S}}_r^{(1)}(1)}^{4} \\
        \frac{1}{f_{2\ell-1}-\alpha_{\overline{\mathcal{S}}_r^{(1)}(2)}} & \frac{1}{f_{2\ell}-\alpha_{\overline{\mathcal{S}}_r^{(1)}(2)}} & 1      & \cdots & \alpha_{\overline{\mathcal{S}}_r^{(1)}(2)}^{4} \\
        \vdots                                                        & \vdots                                                      & \vdots & \vdots & \vdots                                         \\
        \frac{1}{f_{2\ell-1}-\alpha_{\overline{\mathcal{S}}_r^{(1)}(7)}} & \frac{1}{f_{2\ell}-\alpha_{\overline{\mathcal{S}}_r^{(1)}(7)}} & 1      & \cdots & \alpha_{\overline{\mathcal{S}}_r^{(1)}(7)}^{4} \\
    \end{bmatrix}.
\end{align}
Again, its invertibility is guaranteed by the determinant of Cauchy-Vandermonde matrices\cite{Gasca_Martinez_Muhlbach} and the fact that the constants $(f_{2\ell-1}, f_{2\ell}, \alpha_1,\alpha_2,\cdots,\alpha_8)$ are distinct for all $\ell\in[3\xi]$. Recall that $\overline{\mathcal{S}}_r^{(1)}(1)$, $\overline{\mathcal{S}}_r^{(1)}(2)$, etc., refer to distinct elements of $\overline{\mathcal{S}}_r^{(1)}$, i.e., available servers during the first read phase, arranged in ascending order. Therefore, the user is able to reconstruct the desired submodel $\mathbf{W}^{(0)}_{\theta_1}$.

To update the desired submodel with the increment $\mathbf{\Delta}_1$, User $1$ constructs the ACSA increment $\mathbf{P}_n^{(1,\theta_1)}$ as follows.
\begin{align}
    \mathbf{P}_{n}^{(1,\theta_1)}=\diag\left(\widetilde{\Delta}_1^{(1)}\mathbf{I}_{K},\widetilde{\Delta}_2^{(1)}\mathbf{I}_{K},\cdots,\widetilde{\Delta}_{L}^{(1)}\mathbf{I}_{K}\right),
\end{align}
where for all $\ell\in[L]$,
\begin{align}
    \widetilde{\Delta}_\ell^{(1)}=\frac{1}{\alpha_n-f_\ell}\Delta_{\ell}^{(1)}+\dddot{Z}_{\ell,1}^{(1)}.
\end{align}
Note that the ACSA unpacker at time $t=1$ is $\mathbf{\Upsilon}_n^{(1)}=\mathbf{I}_{KL}$ for all $n\in[N]$, therefore,
\begin{align}
     & \mathbf{\Omega}_n^{(1)}\mathbf{\Upsilon}_n^{(1)}\mathbf{P}_{n}^{(1,\theta_1)}\mathbf{Q}_{n}^{(1,\theta_1)}\notag \\
     & =\left[\begin{array}{l}
            \left(\frac{\prod_{i\in\mathcal{S}_w^{(1)}}(\alpha_n-\alpha_i)}{\prod_{i\in\mathcal{S}_w^{(1)}}(f_{1}-\alpha_i)}\right)\left(\frac{1}{\alpha_n-f_1}\Delta_{1}^{(1)}+\dddot{Z}_{1,1}^{(1)}\right)\left(\mathbf{e}_K(\theta_t)+(\alpha_n-f_{1})\ddot{\mathbf{Z}}_{1,s}^{(1)}\right) \\
            \hdashline[2pt/2pt]
            \left(\frac{\prod_{i\in\mathcal{S}_w^{(1)}}(\alpha_n-\alpha_i)}{\prod_{i\in\mathcal{S}_w^{(1)}}(f_{2}-\alpha_i)}\right)\left(\frac{1}{\alpha_n-f_2}\Delta_{2}^{(1)}+\dddot{Z}_{2,1}^{(1)}\right)\left(\mathbf{e}_K(\theta_t)+(\alpha_n-f_{2})\ddot{\mathbf{Z}}_{2,s}^{(1)}\right) \\
            \hdashline[2pt/2pt]
            \multicolumn{1}{c}{\vdots}                                                                                                                                                                                                                                                        \\
            \hdashline[2pt/2pt]
            \left(\frac{\prod_{i\in\mathcal{S}_w^{(1)}}(\alpha_n-\alpha_i)}{\prod_{i\in\mathcal{S}_w^{(1)}}(f_{L}-\alpha_i)}\right)\left(\frac{1}{\alpha_n-f_L}\Delta_{L}^{(1)}+\dddot{Z}_{L,1}^{(1)}\right)\left(\mathbf{e}_K(\theta_t)+(\alpha_n-f_{L})\ddot{\mathbf{Z}}_{L,s}^{(1)}\right)
        \end{array}\right]                                                                         \\
     & =\left[\begin{array}{l}
            \frac{1}{\alpha_n-f_{1}}\Delta_{1}^{(1)}\mathbf{e}_K(\theta_t)+\sum_{i\in[4]}\alpha_n^{i-1}\dot{\mathbf{I}}_{1,i}^{(1)} \\
            \hdashline[2pt/2pt]
            \frac{1}{\alpha_n-f_{2}}\Delta_{2}^{(1)}\mathbf{e}_K(\theta_t)+\sum_{i\in[4]}\alpha_n^{i-1}\dot{\mathbf{I}}_{2,i}^{(1)} \\
            \hdashline[2pt/2pt]
            \multicolumn{1}{c}{\vdots}                                                                                              \\
            \hdashline[2pt/2pt]
            \frac{1}{\alpha_n-f_{L}}\Delta_{L}^{(1)}\mathbf{e}_K(\theta_t)+\sum_{i\in[4]}\alpha_n^{i-1}\dot{\mathbf{I}}_{L,i}^{(1)}
        \end{array}\right].
\end{align}
Thus by the update equation \eqref{eq:update}, the ACSA scheme updates the storage at time $0$, i.e., it holds that $\mathbf{S}_n^{(0)}+\mathbf{\Omega}_n^{(1)}\mathbf{\Upsilon}_n^{(1)}\mathbf{P}_{n}^{(1,\theta_1)}\mathbf{Q}_{1}^{(1,\theta_1)}=\mathbf{S}_n^{(1)}$. Besides, it is guaranteed that $\mathbf{S}_n^{(1)}=\mathbf{S}_n^{(0)}$ for all $n\in\mathcal{S}_w^{(1)}$. This is the end of the full cycle of ACSA-RW at time $1$. We note that in the limit as $L/K\rightarrow\infty$, we have $D_t=7/2=3.5$ and $U_t=6$.

Now let us assume that at time $t=2$, i.e., for the second cycle with a second user, we have $|\mathcal{S}_r^{(1)}|=2, |\mathcal{S}_w^{(1)}|=1$. Therefore, $R_r^{(2)}=1, R_w^{(2)}=2, \#_r^{(2)}=6\xi, \#_w^{(2)}=3\xi$.
Upon receiving the queries, the servers $n, n\in\overline{\mathcal{S}}_r^{(2)}$ respond to the user with the answers as follows.
\begin{align}
    A_n^{(2,\theta_2)}=\left(\mathbf{A}_{n,\ell}^{(2,\theta_2)}\right)_{\ell\in[L]}, n\in\overline{\mathcal{S}}_r^{(2)},
\end{align}
where for all $\ell\in[L]$, we have
\begin{align}
    \mathbf{A}_{n,\ell}^{(2,\theta_2)}=\frac{W_{\theta_2}^{(1)}(\ell)}{\alpha_n-f_{\ell}}+\sum_{i\in[5]}\alpha_n^{i-1}\ddot{I}_{\ell,i}^{(2)}.
\end{align}
Thus, User $2$ recovers the desired submodel by inverting the matrix $\mathbf{C}$ as defined in \eqref{eq:cvmat}. To update the submodel with the increment $\mathbf{\Delta}_2$, User $2$ constructs the ACSA increment $\mathbf{P}_n^{(2,\theta_2)}$ as follows.
\begin{align}
    \mathbf{P}_{n}^{(2,\theta_2)}=\diag\left(\widetilde{\Delta}_1^{(2)}\mathbf{I}_{2K},\widetilde{\Delta}_2^{(2)}\mathbf{I}_{2K},\cdots,\widetilde{\Delta}_{3\xi}^{(2)}\mathbf{I}_{2K}\right),
\end{align}
where for all $\ell\in[3\xi]$,
\begin{align}
    \widetilde{\Delta}_{\ell}^{(2)}=\frac{1}{\alpha_n-f_{2\ell-1}}\Delta_{2\ell-1}^{(2)}+\frac{1}{\alpha_n-f_{2\ell}}\Delta_{2\ell}^{(2)}+\dddot{Z}_{\ell,1}^{(2)}.
\end{align}
Recall that the ACSA unpacker at time $t=2$ is defined as follows.
\begin{align}
    \mathbf{\Upsilon}_n^{(2)} & =\diag\left(\left(\frac{\prod_{i\in\mathcal{F}_1^{(2)}}(\alpha_n-f_i)}{\prod_{i\in\mathcal{F}_1^{(2)}}(f_1-f_i)}\right)\mathbf{I}_K,\cdots,\left(\frac{\prod_{i\in\mathcal{F}_L^{(2)}}(\alpha_n-f_i)}{\prod_{i\in\mathcal{F}_L^{(2)}}(f_L-f_i)}\right)\mathbf{I}_K\right) \\
                              & =\diag\left(\left(\frac{\alpha_n-f_{2}}{f_1-f_2}\right)\mathbf{I}_K,\left(\frac{\alpha_n-f_{1}}{f_2-f_1}\right)\mathbf{I}_K,\left(\frac{\alpha_n-f_{4}}{f_3-f_4}\right)\mathbf{I}_K,\left(\frac{\alpha_n-f_{3}}{f_4-f_3}\right)\mathbf{I}_K,\cdots\right.\notag           \\
                              & \hspace{2cm}\left.\cdots,\left(\frac{\alpha_n-f_{L}}{f_{L-1}-f_L}\right)\mathbf{I}_K,\left(\frac{\alpha_n-f_{L-1}}{f_L-f_{L-1}}\right)\mathbf{I}_K\right).
\end{align}
Therefore, for all $n\in[N]$, we can write
\begin{align}
     & \mathbf{\Upsilon}_n^{(2)}\mathbf{P}_{n}^{(2,\theta_2)}\notag                                                                                                                                                                                                       \\
     & =\diag\left(\left(\frac{1}{\alpha_n-f_{1}}\Delta_{1}^{(2)}+\sum_{i\in[2]}\alpha_n^{i-1}\dddot{I}_{1,i}^{(2)}\right)\mathbf{I}_K,\cdots,\left(\frac{1}{\alpha_n-f_{L}}\Delta_{1}^{(2)}+\sum_{i\in[2]}\alpha_n^{i-1}\dddot{I}_{L,i}^{(2)}\right)\mathbf{I}_K\right).
\end{align}
Accordingly, the second term in the RHS of the update equation \eqref{eq:update} can be written as follows.
\begin{align}
     & \mathbf{\Omega}_n^{(2)}\mathbf{\Upsilon}_n^{(2)}\mathbf{P}_{n}^{(2,\theta_2)}\mathbf{Q}_{n}^{(2,\theta_2)}\notag \\
     & =\left[\begin{array}{l}
            \frac{1}{\alpha_n-f_{1}}\Delta_{1}^{(2)}\mathbf{e}_K(\theta_t)+\sum_{i\in[4]}\alpha_n^{i-1}\dot{\mathbf{I}}_{1,i}^{(2)} \\
            \hdashline[2pt/2pt]
            \frac{1}{\alpha_n-f_{2}}\Delta_{2}^{(2)}\mathbf{e}_K(\theta_t)+\sum_{i\in[4]}\alpha_n^{i-1}\dot{\mathbf{I}}_{2,i}^{(2)} \\
            \hdashline[2pt/2pt]
            \multicolumn{1}{c}{\vdots}                                                                                              \\
            \hdashline[2pt/2pt]
            \frac{1}{\alpha_n-f_{L}}\Delta_{L}^{(2)}\mathbf{e}_K(\theta_t)+\sum_{i\in[4]}\alpha_n^{i-1}\dot{\mathbf{I}}_{L,i}^{(2)}
        \end{array}\right],
\end{align}
which guarantees the correctness of the update. It is remarkable that by the construction of ACSA null-shaper,  for all $n\in\mathcal{S}_w^{(2)}$, it is guaranteed that $\mathbf{\Omega}_n^{(2)}=\mathbf{0}$, thus we have $\mathbf{S}_n^{(2)}=\mathbf{S}_n^{(1)}$ for all $n\in\mathcal{S}_w^{(2)}$. Therefore, the write dropout constraint is satisfied. It is easy to verify that the privacy and security constraints are satisfied, and in the limit as $L/K\rightarrow\infty$, the normalized upload cost is $U_2=7/2=3.5$, the normalized download cost is $D_2=6$. We can see that when the number of dropout servers is decreased, the communication efficiency is accordingly improved (e.g., $U_2<U_1, D_2>D_1$).

\section{Conclusion}
Inspired by the recent interest in $X$-secure $T$-private federated submodel learning, we explored the fundamental problem of privately reading from and writing to a distributed and secure database. By interpreting the private read and write operations as  secure matrix multiplications (between query vectors and stored data), and recognizing that CSA codes are natural solutions to such problems, we constructed a novel Adaptive CSA-RW scheme. ACSA-RW achieves synergistic gains from the joint design of private read and write operations because the same one hot vector representation of the desired message index needs to be secret shared for both the private read and write operations. In addition to allowing private read and write, ACSA-RW also provides elastic resilience against server dropouts, up to thresholds that are determined by the number of redundant storage dimensions. Surprisingly, ACSA-RW is able to fully update the distributed database even though the database is only partially accessible due to write-dropout servers. This is accomplished by exploiting the redundancy that is  already required for secure storage. The scheme allows a memoryless operation of the database in the sense that the storage structure is preserved and users may remain oblivious of the prior history of server dropouts. A promising direction for future work is to explore applications of this idea to multi-version coding \cite{Wang_Multi_Version}. 

\bibliography{Thesis}
\bibliographystyle{IEEEtran}
\end{document}